\documentclass[a4paper,10pt]{article}

\usepackage[T1]{fontenc}
\usepackage[english]{babel}
\usepackage{amsmath, amssymb, graphics}
\usepackage{subfigure}
\usepackage{latexsym}
\usepackage{amssymb}
\usepackage{amsfonts}
\usepackage{theorem}
\usepackage{graphicx}
\usepackage{amscd}
\usepackage[normalem]{ulem}
\usepackage{float}
\usepackage{color}
\usepackage[arrow, matrix, curve]{xy}
\usepackage{natbib}
\bibpunct{(}{)}{;}{a}{}{,}

\newcommand{\mathsym}[1]{{}}

\newcommand{\vect}[1]{{\boldsymbol#1}}

\def\N{\mathbb{N}}

\def\R{\mathbb{R}}
\def\Z{\mathbb{Z}}

\newtheorem{thm}{Theorem}[section]
\newtheorem{proposition}[thm]{Proposition}
\newtheorem{lemma}[thm]{Lemma}
\newtheorem{corollary}[thm]{Corollary}
\newtheorem{definition}[thm]{Definition}
\newtheorem{remark}[thm]{Remark}

\newtheorem{example}[thm]{Example}

\newenvironment{proof}{\textbf{Proof :\newline}}{$\square$}

\begin{document}
\title{Pattern formation in auxin flux}
\author{ Chrystel Feller, Jean-Pierre Gabriel, Christian Mazza\thanks{Corresponding author, D\'epartement de Math\'ematique, Universit\'e de Fribourg, Chemin du Mus\'ee 23, CH-1700 Fribourg, Suisse, christian.mazza@unifr.ch}  and Florence Yerly \protect\hspace{1cm}}
\maketitle

\begin{abstract}
The plant hormone auxin is fundamental for plant growth, and its spatial distribution
in plant tissues is critical for plant morphogenesis. We consider a leading model of the polar auxin flux,
and study in full detail the stability of the possible equilibrium configurations.
We show that the critical states of the auxin transport process
are composed of basic building blocks, which are isolated in a background of auxin depleted cells,
and are not geometrically regular in general. 
The same model was considered recently through a continuous limit and
a coupling to the von Karman equations, to model the interplay of biochemistry
and mechanics during plant growth. 
 Our conclusions might be of interest in
this setting, since, for example, we establish the existence of Lyapunov functions for the auxin flux,
proving in this way the convergence of pure transport processes toward the set of critical configurations.
\end{abstract}
\section{Introduction\label{s.Introduction}}

The plant hormone auxin plays a fundamental role in plant development \citep{Reinhardt,Reinhardt2}, and its  spatial distribution
in plants tissues is critical for plant morphogenesis. Auxin accumulation is spatially 
localized in specific set of cells, where it induces the emergence of new primordia
\citep{Reinhardt}.
 A fundamental
problem consists in understanding how such auxin maxima appear, and how they induce
 the regular pattern observed in plants (see e.g. \citet{Traas}).
On the other hand, experiments show that phyllotaxis strongly depends on the plant physical 
properties, more precisely on elasticity \citep{Green,Dumais1,Dumais2}, and physical forces
provide information for plant patterning \citep{Traas}. Basically, turgor pressure
induces stress, which is  related to  the associated deformation or strain through
Young constants: see e.g. \citet{Boudaoud} where these
notions are explained in the context of plant growth.
 Experiments have shown that lowering the stiffness of cell walls  in the meristem
leads to the emergence of new primordia  \citep{Hamant2}. However,  the interactions between physics-based  and biochemical control
of phyllotaxis is still poorly understood.

Recently, new biologically plausible mathematical models of auxin transport have been proposed \citep{Barbier,Heisler,Jonsson,Smith1}, each of them being able to reproduce some aspects of  phyllotaxis in simulations.
 New mathematical models were also proposed for the 
plant mechanics \citep{Mjolsness}, and for the interaction between mechanics and
biochemistry \citep{Shipman,Newell}. In the latter, the authors use the model for the polar auxin flux proposed
in \citet{Jonsson} for modelling the stress field in their mechanical model.
It should be stressed that all these models are based on hypotheses that have not been
verified experimentally; however they provide new scenari for
understanding plant growth that can be tested experimentally.

\begin{figure}[h!]
\centering
\includegraphics[height=5cm]{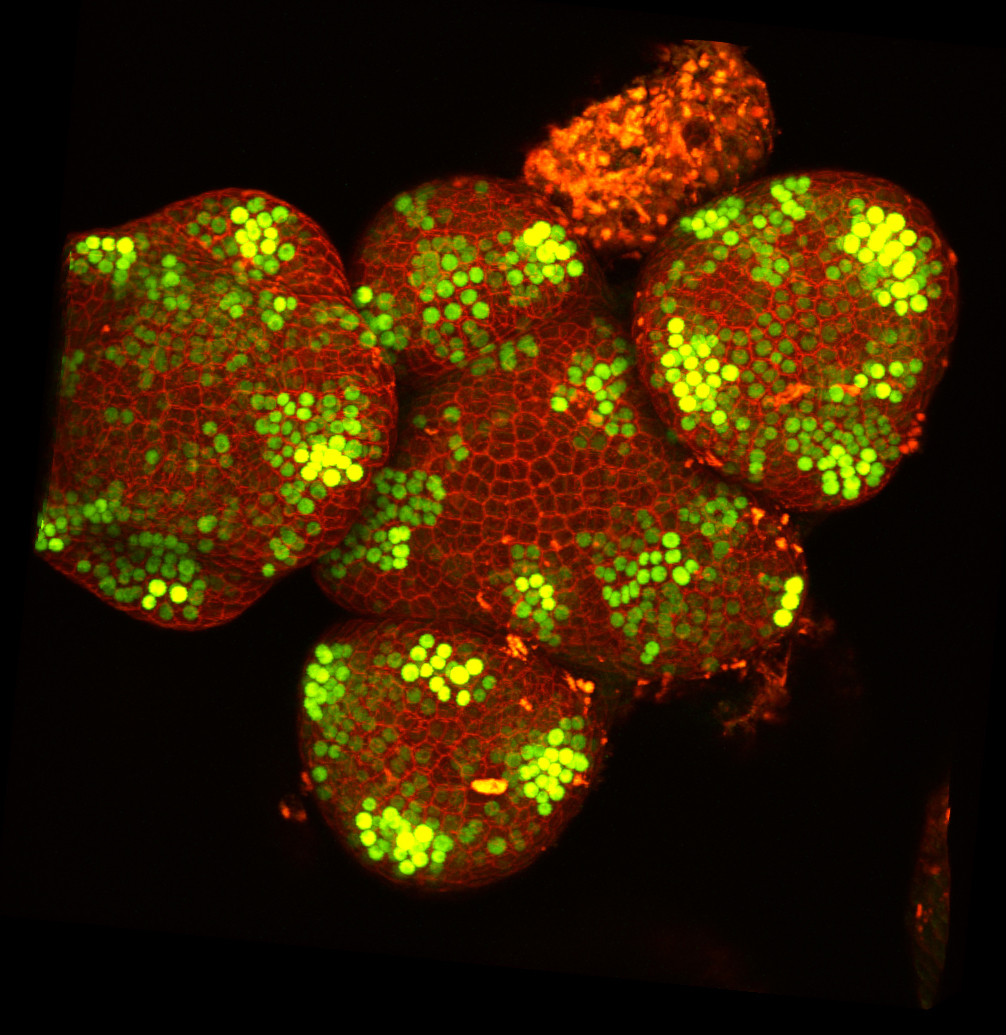}
\caption{\rm{Inflorescence shoot apical meristem of \textit{Arabidopsis thaliana}.
Zones with high auxin concentration
are highlighted by the fluorescent yellow signal auxin reporter from DR5::YFP.
The red signal is highlighting   cell walls stain, using propidium iodide.
}\label{phyllo}}
\end{figure}

Auxin occurs in various plant tissues, where it is
transported by polar cellular transport
in various directions and can explain developmental patterning
phenomena such as vein formation, see e.g. \citet{Scarpella} or \citet{Bayer}.

 In the following, we  consider the  models in \citet{Jonsson} and \citet{Smith1}, based on polar auxin flux.
 Polar auxin flux results from uneven accumulation of the auxin transport regulator PIN in cell membranes.
An essential component is a positive feedback
between auxin flux and PIN localization, resulting
in the reinforcement of polar auxin transport to dedicated 
 routes which develop into vascular tissues.
  We will not enter here into these considerations, but focus on simple models of transport processes (see e.g. the discussion in \citet{Jonsson} and \citet{Shipman}), where a quasi-equilibrium is assumed for PIN proteins. The molecules
present in some cell $i$ may be transported to any neighbouring cell $j$, but they are preferentially transported to the neighbours with the highest auxin concentrations.

Traditionally, models of patterning and morphogenesis have used reaction-diffusion theory. 
Turing demonstrated how, under some hypotheses, the regular patterns observed in phyllotaxis
can be predicted \citep{Turing}. He showed that a combination
of diffusion and a chemical reaction could give rise to
regular patterns. Interesting models are described in \citet{Meinhardt,Thornley} which can, under some hypotheses, predict
phyllotactic patterns.
 As stated previously, the auxin flux is strongly polarized, a phenomenon that cannot be described with
reaction-diffusion models. The recent mathematical models given in \citet{Barbier,Jonsson,Smith1}
are based on transport processes.
 Mathematically, mass transport processes are not well understood, and their study is a challenging problem. 
 We propose here a mathematical study of related dynamical systems. We focus on their critical points and
analyse their geometrical structure and stability.

 Besides stable auxin peaks, the model generates
intervening areas of auxin depletion, as
it is observed experimentally. These auxin depleted sites
reflect an indirect repulsion mechanism since auxin molecules
diffusing through the tissue will be attracted
to the peaks, and diverted from the depleted areas. This
idea of repulsion or spacing mechanism was already considered
a long time ago  \citep{Hofmeister}.

The auxin flux is present everywhere in the plant, so that we choose to describe 
the various plant cells as a connected graph $(\Lambda,E)$. The node set $\Lambda$
represents the cells and $E$ the set of edges. Any edge 
$e=(i\to j)$, $i$, $j\in\Lambda$ indicates that some auxin molecule can
move from cell $i$ to cell $j$. This graph is undirected, and we write
$i\sim j$ to denote that cells $i$ and $j$ are nearest neighbours, so that
auxin can move from cell $i$ to cell $j$, at some rate
$q_{ij}$. 
These transition rates are not well understood at present time and one must rely on  simple models. They
should  capture the fact that an auxin molecule present in some cell
$i$ has the tendency to move to a cell $j\sim i$ when the concentration $a_j$
of auxin molecules present in cell $j$ is high. The simplest model  accounting
for this idea is given by \citet{Jonsson}
\begin{equation}\label{Rates}
q_{ij} = \frac{a_j}{\kappa + \sum_{k\sim i}a_k},
\end{equation}
for some positive constant $\kappa$, which is of Michaelis-Menten or Monod type. 
Let $L=\vert\Lambda\vert$ be the number of cells.
In the model given in \citet{Jonsson} (see also \citet{Smith1,Sahlin}), $a_i(t)$, for $i=1,\cdots , L$, denotes the concentration or the number of auxin molecules
in cell $i$ at time $t$, and is assumed to evolve according to the differential equations
\begin{equation}\label{eq:jonsson}
\frac{\mathrm{d}a_i}{\mathrm{dt}} = f_i(\vect{a}) = D \sum_{k \sim i} (a_k -a_i) + T\sum_{k \sim i} \Big ( a_k \underbrace{\frac{a_i}{\kappa + \sum_{j \sim k} a_j }}_{=q_{ki}(\vect{a})}  - a_i \underbrace{\frac{a_k}{\kappa + \sum_{j \sim i} a_j }}_{=q_{ik}(\vect{a})} \Big),
\end{equation}
 for $i=1, \ldots, L$.   The term $a_i q_{ik}$ gives the mean number of auxin molecules moving
from cell $i$ to cell $k$ per unit time, and $D \sum_{k \sim i} (a_k -a_i)$ is a diffusive part, usually assumed to 
be weak with a small diffusion coefficient $D$. The second term corresponds
to the mass transport process, which is known to be the main actor of the
patterning process in plants. One can add auxin production and degradation
terms, but, there is no clear biological evidence
about where auxin is produced, and experiments show that it is
not produced in the meristem, but imported from the leaves \citep{Reinhardt,Reinhardt3}.

\subsection{Results}

Direct quantitative measurements of auxin distribution in plant tissues
are very difficult due to the small size of the meristematic tissues at the time
of patterning. Therefore, biologists rely on indirect markers based
on auxin-regulated genes that encode fluorescent proteins.
Figure \ref{phyllo} shows a typical output, where
domains rich in auxin appear as regions of strong green fluorescence.
The pattern is quite noisy; this might be due either to
the  indirect experiments, or to 
the fact that the number of auxin molecules is not too high.
 (\ref{eq:jonsson}) might model the limiting behavior
of this random particle system when the number of molecules  tends to infinity.
We introduce such a particle system in Section \ref{Stochastic} and  justify 
equations like (\ref{eq:jonsson}) using law of large numbers.

We then focus  on
the properties of (\ref{eq:jonsson}), like the non-negativity of
the solutions (see Proposition \ref{Invariance}). This dynamical system can be written in the compact form 
$$\frac{\mathrm{d}\vect{a}}{\mathrm{dt}}=f(\vect{a}),$$
where $\vect{a}(t) = (a_i(t))_{1\le i\le L}$ is the vector of auxin concentrations.
 The related critical points are the vectors $\vect{a}^*$ satisfying  $f(\vect{a}^*)=0$. They are the candidates for describing the equilibrium auxin concentrations. For example,
$a_i =0$ means that there is (almost) no auxin molecules in cell $i$, while a subset of cells
$J$ such that $a_j >0$ for $j\in J$ indicates a hot spot which might correspond to an auxin peak.

The critical points play a  fundamental role in the dynamic, and one can suspect that any solution $\vect{a}(t)$ of (\ref{eq:jonsson}) will approach
such critical points as $t$ is large. Of course, this is wrong for general dynamical systems, but here, the model
is supposed to catch pieces of biological reality, and 
the robustness of the  regular geometries observed in plants
suggests that this might well be the case.
Some of these critical points are repulsive
or unstable, that is, the orbits or the solutions of (\ref{eq:jonsson}) will avoid them. In the contrary, some of them
will be attractive. Given a critical point $\vect{a}^*$, a mathematical way of checking the stability or the unstability
of $\vect{a}^*$ is to compute the Jacobian ${\rm d}f(\vect{a}^*)$, by retaining only its spectrum, that is the
set of all eigenvalues of ${\rm d}f(\vect{a}^*)$. For example, $\vect{a}^*$ is unstable when there is
an eigenvalue having a positive real part. 

\begin{definition}
We say that a critical point $\vect{a}$ is stable when all the eigenvalues 
of the Jacobian evaluated at $\vect{a}$ have non-positive real parts.
\end{definition}

Section \ref{s.CriticalPoints} is concerned with the characterization of the set of critical points,
mainly focusing on  pure transport processes.

For $D=0$, we first consider critical points  $\vect{a}>0$, meaning that $a_i > 0$ for all $i$.  Corollary \ref{CriticalAdjacency} shows that such elements are precisely the positive solutions of the linear equation
\begin{equation}\label{Linear}
\Gamma \vect{a} = c \ {\bf 1},\ c\text{ constant},
\end{equation}
where $\Gamma$ is the adjacency matrix of the graph $G$, with entries $\Gamma_{ij}\in \{0,1\}$ such
that $\Gamma_{ij} = 1$ if and only if cells $i$ and $j$ are nearest neighbours, and
${\bf 1}$ is the vector having all components equal to 1.

Next, we focus on critical points such that 
 $a_i =0 $ for $i$ belonging to some subset $I\subset \Lambda = \{1,\cdots, L\}$.
They correspond to auxin depleted cells. The graph decomposes into a product
of sub-graphs $\gamma$, which are the connected components of the sub-graph of $G$ induced
by the node set $J=\Lambda\setminus I$.
We thus look for $\vect{a}$ having positive components $a_j > 0$ for $j\in J$,
 which should correspond in some sense to auxin peaks.
 We obtain the distribution of auxin in such components, denoted by
$\vect{a}\vert_\gamma$, by solving the linear systems
${\Gamma_\gamma \vect{a}\vert_\gamma = c_\gamma {\bf 1}\vert_\gamma}$. A typical example of such configurations is given in Figure \ref{stable}, where the elements of
$I$ are black and the various components $\gamma$ red.

We then turn to the asymptotic behavior of the solutions of system (\ref{eq:jonsson}), and establish in Proposition \ref{ConvergenceAuxin} that
every solution converges toward the set of critical points. Our technique is based on Lyapunov functions, that is, we look for
a function $H(\vect{a})$ which should be decreasing along the orbits of (\ref{eq:jonsson}), like energy in physics. We proved that, for pure transport processes with $D=0$, the function
$$H(\vect{a})= -\kappa \langle{\bf 1},\vect{a}\rangle -\frac{1}{2}\langle\vect{a},\Gamma \vect{a}\rangle,$$
where $\langle\cdot,\cdot\rangle$ denotes the scalar product, satisfies
$$\frac{{\rm d}H(\vect{a}(t))}{{\rm d}t}\le 0$$
for any solution of (\ref{eq:jonsson}).  \citet{Newell} also considered the differential system (\ref{eq:jonsson}) by taking 
a spatial continuous limit, and showed that the limiting 
equation is a p.d.e.  similar to the von Karman equations from nonlinear elasticity theory: 
$$\frac{\partial w}{\partial t} = \triangle^2 w + P \triangle w + \mathrm{const} \cdot w + \text{ nonlinear terms}.$$
The von Karman equations are of gradient type (see e.g. \citet{Shipman}), where the potential is given by the elastic energy.
These energy functionals were then used in \citet{Newell} and \citet{Newell2} to 
provide a very interesting mechanical explanation of the appearance of Fibonacci numbers in plant patterns based
on buckling.
However, 
the limiting equations associated with the auxin flux are not of gradient type, see the discussion in \citet{Newell}. For the basic
dynamical system (\ref{eq:jonsson}), our result shows that the system is minimizing
the energy $H$, without being of gradient type.

Section \ref{SpecialClass} considers stability, and   
Proposition \ref{StabilityGeneral} shows that the Jacobian
 ${\rm d}f(\vect{a}) = (\partial f_i/\partial a_j)$ evaluated at $\vect{a}$ is  given by
 $${\rm d}f(\vect{a})=\frac{1}{N^2}d(\vect{a})\Gamma \Big(c\ {\rm id}-d(\vect{a})\Gamma\Big),$$
 where $d(\vect{a})$ is the diagonal matrix of diagonal given by
 $\vect{a}$. This permits to check the stability of the critical points for various graphs.
We present various
 results on graphs of interest for plant patterning questions, like the circle or the two-dimensional grid.
 As stated previously, the positive solutions $\vect{a}\vert_\gamma$ to the linear system
 $\Gamma_\gamma \vect{a}\vert_\gamma = c_\gamma {\bf 1}\vert_\gamma$
 provide restrictions of the critical points to the connected components $\gamma$.
 We give
 a particularly simple condition on the sub-graph 
 $\gamma$ of $G$ induced by the set $J=\Lambda\setminus I$
  ensuring  the non-stability of  $\vect{a}\vert_\gamma$. Let
 ${\cal N}_i$, $i\in \Lambda$ be the neighbourhood of $i$, that is the set of nodes $j$
 such that $j\ne i$ and $j\sim i$.
  The configuration $\vect{a}\vert_\gamma$  is unstable
 when the sub-graph  $\gamma$ contains a path of length 4, of the form 
 $$i_0 \to i_1 \to i_2 \to i_3,$$
  such that
 $$i_1 \in {\cal N}_{i_0},\ i_2 \not\in {\cal N}_{i_0} \text{ and } i_3\not\in {\cal N}_{i_0}.$$

 \begin{figure}[h!]
 \centering
 \includegraphics[width=6cm]{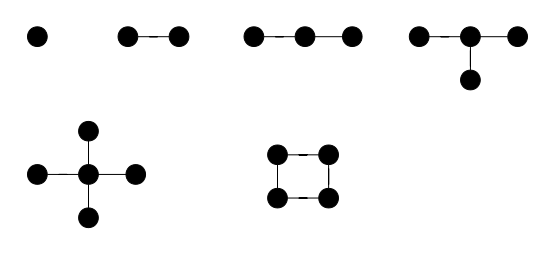}
\caption{ Example of components $\gamma$ of the two-dimensional grid
that can potentially yield stable configurations, see
Corollary \ref{Unstable}. \label{subgraph}}
\end{figure}

For example, if $G$ is a two-dimensional grid, any stable configuration is composed of patches of
the basic building blocks given in Figure \ref{subgraph}; These patterns are however not geometrically regular in general, see Figure \ref{stable}. The more involved model of \citet{Smith1}, which uses PIN proteins in
a direct way (here we assume a quasi-equilibrium, see \citet{Jonsson}), produces more regular patterns in simulations.
In this setting, the transition rates are forced to follow exponential distributions. Hence, a strong selection
based on rates of the form $\exp( b a_i)$, $b>0$ instead of the linear function $a_i$ seems to regularize the
critical points. 
Of course, it might be interesting to justify such a choice biologically.
 We also argue
in what follows that the critical configurations produced by the auxin flux might be more regular when coupled
to periodic potentials.

\begin{figure}[h!]
\centering
\includegraphics[height=5cm]{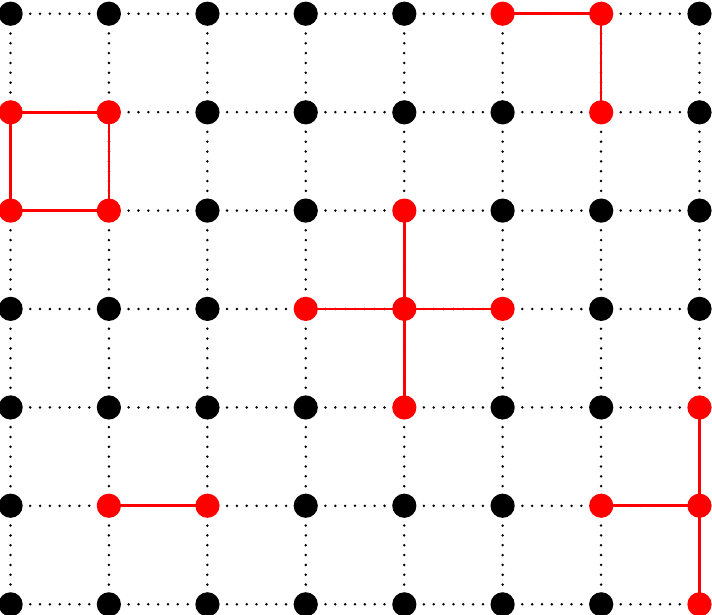}

\caption{\rm A potentially stable configuration when the graph $G$ is a rectangular grid, for the pure transport process.  The black circles
correspond to the values $a_i =0$, $i\in I$ (the set auxin depleted cells), while the red circles are such that $a_i > 0$, corresponding to auxin peaks. One can construct the set of all
stable configurations by playing with the building block given by the square, the star, and the various parts of the star. This shows that dynamical system (\ref{eq:jonsson}) does not  necessarily produce regular patterns. We can however give examples where such configurations are unstable, see Section \ref{StabilityGrid}\label{stable}}
\end{figure}

 It might well be  that the auxin flux self-organize in regular patterns when coupled to mechanical forces, for example, as already
stated in the Introduction, see \citet{Newell}. 
In the same spirit, we introduce a simple model   coupling  the auxin flux
to a potential $\phi$, which might model
deformations, curvature or effects related to the meristem elasticity. We provide an example of the form
\begin{equation}\label{Pot}
 \frac{{\rm d}a_i(t)}{{\rm d}t}=f_i(\vect{a})+\sum_{j\sim i}(a_j \phi_i -a_i \phi_j),
\end{equation}
 $i=1,\ldots,L$. If the potential itself has some regularities, as it is the case
 in specific model given in \citet{Newell}, the auxin flux will exhibit much more regular patterns,
 see e.g. Figure \ref{Fig4}.
 
\begin{figure}
\begin{minipage}{0.5\textwidth}
\begin{center}
\includegraphics[height=4cm]{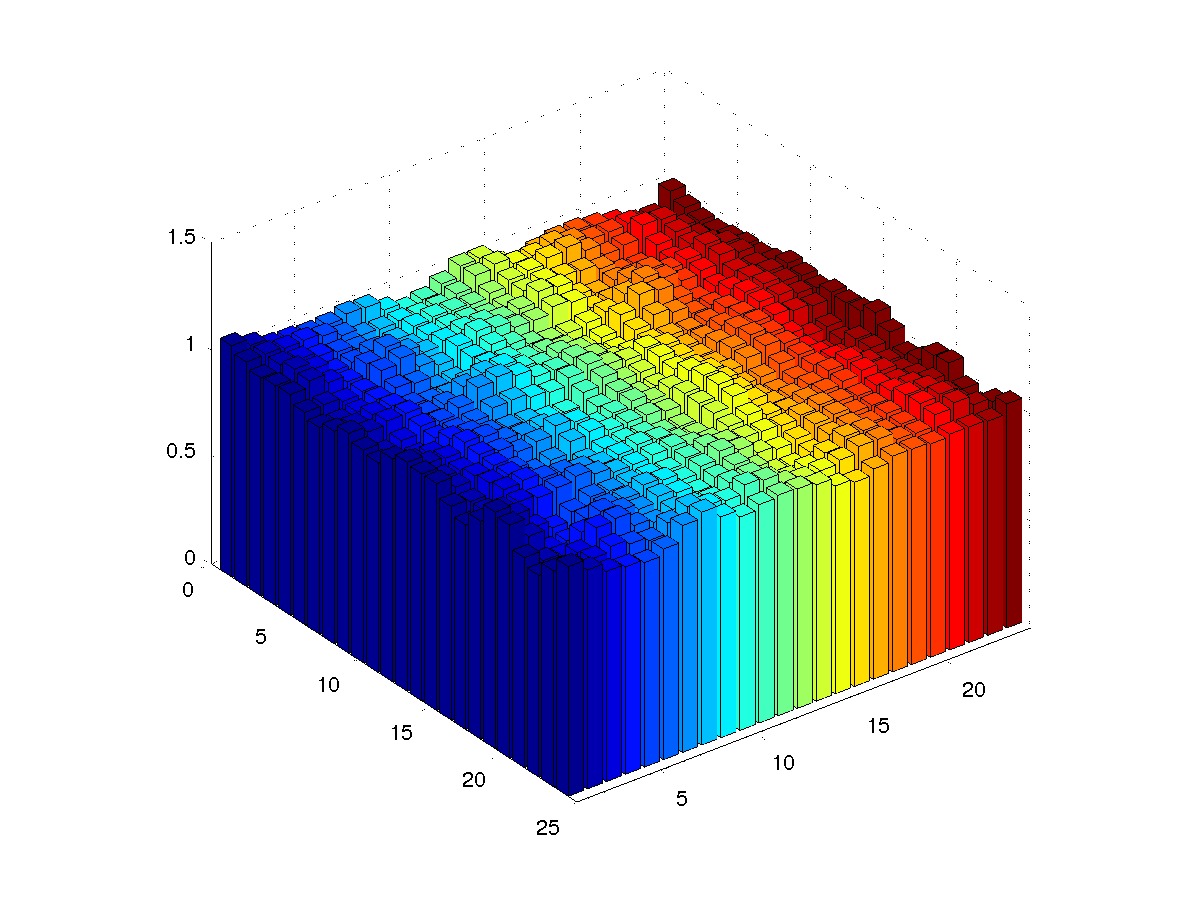}

\vspace{3mm}
(a) $\vect{a}(t)$ for $t\approx 0$
\end{center}
\end{minipage}
\begin{minipage}{0.5\textwidth}
\begin{center}
\includegraphics[height=4cm]{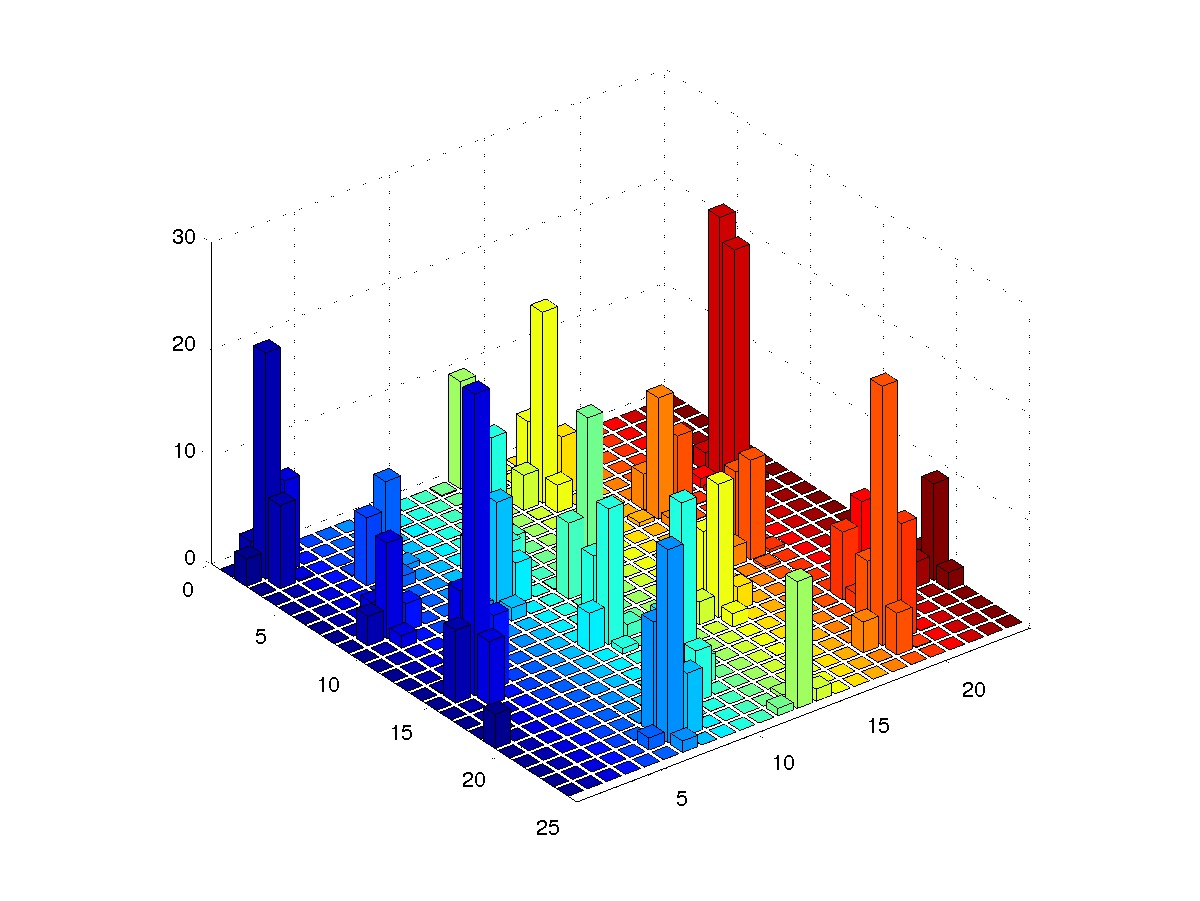}

\vspace{3mm}
(b) $\vect{a}(t)$ for large $t$.
\end{center}
\end{minipage}
\caption{Simulation of the orbits of the differential equation (\ref{Pot}) with $T=1, D=0$ and a potential $\phi(x,y)=
\sin(4\pi x/A)\sin(4\pi y/B)$ on a torus, where
$x=1,\cdots,A$ and $y=1,\cdots,B$. The initial state is flat. (b) shows the state 
$\vect{a}(t)$ for large $t$: one sees regularly spaced auxin peaks, which are isolated
in a background of auxin depleted cells. The potential and transport terms drift thus the process
toward more regular patterns, while the transport process creates domains of auxin depletion.
\label{Fig4}}
\end{figure}

 Finally, the model  provides  an interesting conclusion:  for most graphs, stable configuration
 are composed of building blocks isolated in a sea of auxin depleted cells. This might be the basis 
 for  repulsion between primordia: auxin molecules will not have the tendency to move toward them, leading to
 indirect repulsion. The idea of such repulsive force appeared a long time ago in the work of
 \citet{Hofmeister}. 
  Many authors have
 used this hypothesis to develop very interesting mathematical models, all leading to phyllotactic patterns
 observed in nature, like Fibonacci numbers, the Golden Angle or helical lattices, see
 \citet{Adler,Atela,Douady,Kunz,Levitov}.

\section{A stochastic model of auxin transport \label{Stochastic}}

We consider a stochastic process  related to differential equation 
(\ref{eq:jonsson}), describing the random numbers of auxin molecules
$\eta_t(i)\in\N$ present in cell $i$ at time $t$, $i=1,\cdots, L$.
The state space of this stochastic process is denoted by
$\Omega_L = \N^{\Lambda}$, where
$\Lambda$ is the set of $L$ cells (the nodes of the graph).
Looking at equation (\ref{eq:jonsson}), we
define transitions by supposing that any auxin molecule present in cell $i$ at time $t$
can be transported to a neighboring cell $j$ at rate $\bar q_{ij}(\eta)$
of the form
\begin{equation}\label{StochasticRates}
\bar q_{ij}(\eta)=\frac{\eta (j)}{\bar\kappa +\sum_{k\sim i}\eta (k)},
\end{equation}
when $\eta (i)\ge 1$.
This defines a Markov process with state space $\Omega_L$, describing the stochastic moves of the various
auxin molecules. Let $M$ denote the total number
of molecules.
 It turns out that the ordinary differential
equation (\ref{eq:jonsson}) describes the large $M$ limit of the stochastic process
(weak noise limit). This random particle system is then described 
as a gaussian process $X_M(t)\approx \eta_t /M$ in $\R^L$ drifted by the
solution $\vect{a}(t)$ of (\ref{eq:jonsson}) for some covariance function. This approximation will be mathematically rigourous
if the constants $\kappa$ and $\bar\kappa$ are related in such a way that
$\bar\kappa = M \kappa$, and the limiting behavior of the rescaled
number of auxin molecules is such that $\eta_t (i)/M \approx a_i(t)$, where
$\vect{a}(t)$ solves (\ref{eq:jonsson}), with
$\sum_{i\in\Lambda}a_i(t)\equiv 1$.

Such stochastic particle systems are known as {\it density dependent population processes},
and the above limit has been  treated in detail in \citet{Ethier},
and corresponds to a law of large numbers.
Notice that different kinds of limits can also
be considered. Stochastic mass transport processes of this type have also appeared in physics, and are known as 
{\it generalized zero range processes}, see e.g. \citet{Evans,Godreche,Grosskinsky,Kipnis}. In this setting, hydrodynamical limits are considered, when both $M$ and $L$ tend simultaneously
to $\infty$ in such a way that $M=\rho L$, for a fixed density. Simulations show the appearance of condensates
when $\rho$ is larger than a critical threshold $\rho_c$, which might represent auxin peaks in some way. Mathematically,
the theory of condensation is not developed at present time for these general processes, so that we here
focus on the weak noise limit. \\

The gaussian approximation of $\eta_t /M$ is defined as follows: for $i=1,\cdots, L$, consider
the unit vectors $e_i$ with $e_{i}(j)=0$ when $j\ne i$ and $e_i(i)=1$. Let
${\rm d} f$ be the Jacobian  ${\rm d} f = (\partial f_i/\partial a_j))_{i,j=1,\cdots,L}$.

For simplicity, we illustrate the transition rates for cells arranged along a circle:
the rate functions  are given by functions $\beta_l(\vect{a})$, $l\in \Z^l$, satisfying
\begin{align*} 
\beta_{e_{i+1}-e_i}(\vect{a}) &= (D a_i +T\frac{a_{i+1}a_i}{\kappa +a_{i-1}+a_{i+1}} ) \ ; &\mathrm{for} \; i= 1, \ldots, L, \\
\beta_{e_{i-1}-e_i}(\vect{a}) &= (D a_i +T\frac{a_{i-1}a_i}{\kappa + a_{i-1}+a_{i+1}} ) \ ; &\mathrm{for} \; i= 1, \ldots, L,\\
\beta_{l}(\vect{a}) &=0 & \mathrm{for} \ ; l \neq e_{i-1}-e_{i}, e_{i+1}-e_{i}.
\end{align*}
For example $e_{i+1}-e_i$ means that an auxin molecule of cell $i$ has been
transported in cell $i+1$.  For arbitrary graphs, the definitions of the
rates $\beta_l$ are similar.

With these notations, we can define the matrix $G$
$$G(\vect{a})=\sum_{l\in\Z^L}\beta_l(\vect{a})l l^*,$$
which will be an essential element of the covariance matrix associated with the gaussian approximation.
Consider the following matrix valued differential equation
$$\frac{\partial \phi(t,s)}{\partial t}={\rm d} f(\vect{a}(t))\phi(t,s),\ \ \phi(s,s)={\rm id}.$$
Then, as $M$ is large, one gets that (see e.g. \citep{Ethier})
$$ \frac{\eta_t}{M}= \vect{a}(t) + \frac{1}{\sqrt{M}}V_t,$$
where $V_t$ is a gaussian process  of mean
$\phi(t,0)V(0)$ and of covariance function
$${\rm Cov}(V(t),V(r)) = \int_0^{\min\{t,r\}}\phi(t,s)G(\vect{a}(s))\phi(r,s)^* {\rm d}s.$$

\section{Basic properties of the auxin flux\label{s.BasicProperties} }

\begin{proposition}\label{Invariance}

Every solution $\vect{a}$ of (\ref{eq:jonsson}) starting in $\R^L_{\geq 0}$ remains non-negative, and is conservative, that is,
$$\forall t \in \R_{\geq 0}, \; \sum_i^L a_i(t) =\sum_i^L a_i(0) = \rho L.$$
Moreover, the system (\ref{eq:jonsson}) admits a unique solution defined  over $ [0,+\infty)$. 
When $a_i(0)>0$, then $a_i(t) > 0$, $\forall t > 0$.
For pure transport processes with
$D=0$, $a_i(0)=0 \Rightarrow a_i(t)\equiv 0,\ \forall t > 0$.
\end{proposition}

The proof of proposition \ref{Invariance} is given in Section \ref{Appendix}.

Let us rewrite the system \eqref{eq:jonsson}, for $1\leq i\leq L$

\begin{equation}\label{Diff-ai}
\dot a_i=D\sum_{k\sim i}a_k+T\sum_{k\sim i}(\frac{a_k}{\kappa+\sum_{j\sim k}a_j}-\frac{a_k}{\kappa+\sum_{j\sim i}a_j}-\frac{D}{T})a_i,
\end{equation}
with the initial condition $\vect{a}(0)\in \R_+^L$.

\begin{proposition}
If the graph is connected and $D>0$, the only critical point of \eqref{Diff-ai} in $\R_+^L$ admitting zero components is the origin.
\end{proposition}

\begin{proof}
Let $a_i=0$ where $a_i$ is the $i$-th component of a critical point $a \in \R_+^L$ of \eqref{Diff-ai}. Clearly \eqref{Diff-ai} entails $\sum_{k\sim i}a_k=0$ and the non-negativity of each term, $a_k=0$ for all
$k\sim i$. Since the graph is connected we deduce that $a_k=0$ for all $1\leq k \leq L$.
\end{proof}

\begin{proposition}
Let us assume that the graph is connected and $D>0$. If $\sum_{k=1}^L a_k(0)>0$, then for all $i\in\{1,...,L\}$, we have $\underline \lim_{t\rightarrow +\infty}a_i(t)>0.$
\end{proposition}

To prove the previous proposition, we will use the following Proposition, see \citet{Gabriel}.

\begin{proposition}
Let $f:\R_+\rightarrow \R$ be twice differentiable and bounded together with $\ddot f$. If, as $n\rightarrow +\infty$, $t_n\uparrow +\infty$  and $f(t_n)\rightarrow \underline \lim_{t\rightarrow +\infty} f(t)$ (or $f(t_n)\rightarrow \overline \lim_{t\rightarrow +\infty} f(t)$), then $\dot f(t_n)\rightarrow 0$.
\end{proposition}

\begin{remark}
\begin{enumerate}
\item[(1)] The boundedness of $f$ and $ \ddot f$ implies the one of $\dot f$.
\item[(2)] The assumptions in the preceding proposition can be weakened without changing essentially the proof: "$f:\R_+\rightarrow \R$ be twice differentiable and bounded together with $\ddot f$ " can be replaced by "$f:\R_+\rightarrow \R$ is bounded and differentiable and $\dot f$ is uniformly continuous".
\end{enumerate}
\end{remark}

\begin{proof}
If $a_k(0)=0$ for all $1\leq k\leq L$, then the unique solution is identically zero. Otherwise $\sum_{k=1}^L a_k(0)>0$. Let us suppose that for some $i\in\{1,...,L\}$, $$\underline \lim_{t\rightarrow +\infty}a_i(t)=0.$$ Let us introduce the notation $\underline a_i=\underline\lim_{t\rightarrow +\infty}a_i(t)$. Since $a_i(t)$ is bounded together with its second derivative, the preceding proposition applies and for any sequence $t_n\uparrow +\infty$ such that $a_i(t_n)\rightarrow \underline a_i$, we have $\dot a_i(t_n)\rightarrow 0$ as $n\rightarrow +\infty$.
Every  $a_k(t_n)$ being bounded in the right-hand member of the equation for $\dot a_i(t_n)$, we conclude that $\lim_{n\rightarrow +\infty}D\sum_{k\sim i}a_k(t_n)=0$. The non-negativity of each $a_k(t_n)$ entails $\lim_{n\rightarrow +\infty}a_k(t_n)=0=\underline a_k$ for every ${k\sim i}$. According to the above proposition, $\lim_{n\rightarrow +\infty}\dot a_k(t_n)=0$ for every ${k\sim i}$ and since the graph is connected, repeating the same argument provides  $\lim_{n\rightarrow +\infty}\dot a_j(t_n)=0$ for every $j\in\{1,...,L\}$. Thus $0=\lim_{n\rightarrow +\infty}\sum_{1\leq j\leq L}a_j(t_n)=\sum_{k=1}^L a_k(0)>0$, a contradiction.
\\
\\
As a consequence, for $D>0$, it is impossible to have $\lim_{t\rightarrow +\infty}a_i(t)=0$, and thus none of the compartments can become empty asymptotically.
\end{proof}

\section{Tools from Markov Chain theory\label{s.Tools} }
We will use notions from Markov chain theory, and hence consider generators  $Q:\ \Lambda\ {\rm x}
\ \Lambda \longrightarrow \R$, $Q =\{q_{ij},\ i,\ j\in\Lambda\}$, 
such that
$$q_{ij}\ge 0,\ \text{ for } i\ne j \text{ and } q_{ii}=-\sum_{j\ne i}q_{ij}.$$ 
 For example, the  auxin flux described by (\ref{eq:jonsson}) contains implicitly a generator $Q(D,T,\vect{a})$ given by
 
 \begin{equation}\label{LaplaceBasic}
\begin{cases}
q_{ij}(D,T,\vect{a})= D + T q_{ij}(\vect{a}),& i \sim j, \\
 q_{ij}(D,T,\vect{a})=0 , & i \nsim j, i \neq j. \\
 q_{ii}(D,T,\vect{a})= - \sum_{j\neq i} q_{ij}(D,T,\vect{a}),
\end{cases}
\end{equation}
where we set
$$q_{ij}(\vect{a})= \frac{a_j}{\kappa + \sum_{k \sim i} a_k }.$$

$Q$ is {\bf irreducible} when for any pair of nodes $(i,j)$,  there is a path
$i_0 =i\to i_1\to i_2\to \cdots \to i_k =j$ such that
$q_{i_n i_{n+1}}>0$, $n=0,\cdots, k-1$. When 
$Q$ is irreducible, one can prove that there is a unique {\bf invariant probability
measure} $\pi$ satisfying $\pi^* Q =0$.

An irreducible transition kernel $Q$ of invariant probability measure $\pi$ is said to be {\bf reversible} when 
$$\pi_i q_{ij}\equiv \pi_j q_{ji},\ \forall i\ne j.$$

\section{Characterization of the critical points\label{s.CriticalPoints} }
We can write (\ref{eq:jonsson}) in the more compact form
$$\frac{\mathrm{d}a_i}{\mathrm{dt}}  = f_i(\vect{a})=\sum_{j\sim i} ( a_j q_{ji}(D,T,\vect{a})-a_i q_{ij}(D,T,\vect{a})),
\ \ \frac{\mathrm{d}\vect{a}}{\mathrm{dt}}  = f(\vect{a})= \vect{a}^{\ast}Q(D,T,\vect{a}).$$
Our first aim is to look for the critical points of the above dynamical system, that is,
to find the element $\vect{a}\in\R^L$ solving the equations $f(\vect{a})=0$,
which can be rewritten as
$\vect{a}^{\ast}Q(D,T,\vect{a})=0$. Hence,
 any solution to $f(\vect{a})=0$ is an invariant measure associated with
the transition function $Q(D,T,\vect{a})$. We will use the following facts:
\begin{itemize}
\item{} When $D>0$,  the generator $Q(D,T,\vect{a})$ is irreducible.
\item{} For pure transport processes  where $D=0$ and $T>0$, $Q(0,T,\vect{a})$ is irreducible
if and only if $a_i >0$ $\forall i$.
\end{itemize}
In the irreducible case,  let $\pi(\vect{a})$ denote the associated positive invariant probability measure. 
We thus look for $\vect{a} > 0$ such that
\begin{equation}\label{FundamentalEquation}
\frac{\vect{a}}{\sum_{i\in\Lambda}a_i} =\pi(\vect{a}).
\end{equation}

\subsection{The irreducible case}
\subsubsection{Pure transport processes}
If $Q(0,T,\vect{a})$ is reversible, the equation $f(\vect{a})=0$ is equivalent to the set of equations
\begin{equation}\label{LocalEquation}
a_i q_{ij}(0,T,\vect{a})\equiv a_j q_{ji}(0,T,\vect{a}),\ \ i\ne j.
\end{equation}

 In what follows, we will use the functions
\begin{equation}\label{N}
N_k =N_k(\vect{a}) = \kappa + \sum_{j\sim k}a_j.
\end{equation}

\begin{lemma}\label{ReversibilityPureTransport}
Let $G$ be a connected graph. Assume that $D=0$ and $T>0$. Then
$Q(0,T,\vect{a})$ is reversible $\forall \vect{a}>0$, of invariant probability measure given by
\begin{equation}\label{InvariantMeasureTransport}
\pi (\vect{a}) =\Big(\frac{a_i N_i}{Z(\vect{a})}\Big)_{i\in\Lambda},
\end{equation}
where
$$Z(\vect{a})=\sum_{i\in\Lambda}a_i N_i
=\kappa \sum_{i\in\Lambda}a_i +  \sum_{i \in \Lambda} \sum_{j\sim i}a_i a_j.$$
In this case, $\vect{a}>0$ is a critical point with $f(\vect{a})=0$ if 
and only if $N_i(\vect{a})$ does not depend on $i$, with
\begin{equation}\label{NConstant}
N_i(\vect{a})\equiv \frac{Z(\vect{a})}{\sum_{i\in\Lambda }a_i}=\kappa + \frac{  \sum_{i \in \Lambda} \sum_{j\sim i}a_i a_j}{\sum_{i\in\Lambda}a_i}.
\end{equation}
\end{lemma}
\begin{remark}\label{Reinforced1}
The transition rates $q_{ij}(\vect{a})$ are similar to the rates associated with a  family of Markov chains
used in the study of vertex-reinforced random walks, see  \citet{Benaim1,Benaim2} and \citet{Pemantle}, 
and Lemma \ref{ReversibilityPureTransport} is an adaptation
of these results.
 Interestingly, such
vertex-reinforced random walks are approximated by deterministic dynamical systems called
{\it replicator dynamics}, of the form
$$\frac{{\rm d}a_i}{{\rm d}t}= a_i (N_i'(\vect{a})-H'(\vect{a})),$$
where $N_i'(\vect{a})=N_i(\vect{a})-\kappa$ and $H'(\vect{a})=\sum_{i\in\Lambda}a_i N_i'$.
In this setting, the function $H'$ plays the role of a Lyapunov function.
 We will also find a similar Lyapunov function, see Section \ref{s.Convergence}.
\end{remark}

\begin{proof}
Assume, without loss of generality, that $T=1$.
First notice that
\begin{eqnarray*}
\sum_{j\sim i}\pi(\vect{a})_j q_{ji}(0,T,\vect{a})&=&\sum_{j\sim i}\frac{a_j N_j}{Z(\vect{a})}\frac{a_i}{N_j}\\
  &=&\frac{1}{Z(\vect{a})}\sum_{j\sim i}a_i a_j = \frac{a_i}{Z(\vect{a})}\sum_{j\sim i}a_j
  = \frac{a_i (N_i-\kappa)}{Z(\vect{a})}.
\end{eqnarray*}
The identity
$$\pi(\vect{a})_i q_{ii}(0,T,\vect{a})=-\frac{a_i N_i}{Z(\vect{a})}\sum_{j\sim i}\frac{a_j }{N_i}=-\frac{a_i (N_i-\kappa)}{Z(\vect{a})},$$
shows that
$$\sum_{j\sim i}\pi(\vect{a})_j q_{ji}(0,T,\vect{a})+\pi(\vect{a})_i q_{ii}(0,T,\vect{a})=0,$$
so that $\pi(\vect{a})$ is a invariant probability measure for $Q(0,T,{\bf a})$.
$\bf{a}>0$ is a critical point with $f(\vect{a})=0$ if and only if $\frac{\vect{a}}{\sum_{i \in \Lambda}a_i}$ is an invariant measure for $Q(0,T,{\bf a})$.
Because of the unicity of the invariant measure, we obtain
$$N_i(\vect{a})\equiv \frac{Z(\vect{a})}{\sum_{i\in\Lambda }a_i}.$$

\end{proof}

 Let $\Gamma$ be the adjacency matrix of the
graph $G=(\Lambda,E)$, that is, the matrix with entries given by
$\Gamma_{ij}= 1$, when $i\ne j$ and $i\sim j$, and
$\Gamma_{ij}=0$ otherwise. We summarize the above results in the following

\begin{corollary}[Pure Transport Processes]\label{CriticalAdjacency}
Assume that $D=0$ and $T>0$ (no diffusion), and
consider only positive $\vect{a}>0$.
 Then,
\begin{equation}\label{SolAdjacency}
f(\vect{a})=0 \hbox{ if and only if }\Gamma \vect{a} = c(\vect{a}) {\rm\bf 1},\ {\rm\bf 1}=(1,\cdots,1)^*,
\end{equation}
where
\begin{equation}\label{C}
c(\vect{a})=\frac{ \sum_{i \in \Lambda} \sum_{j \sim i}a_i a_j}{\sum_{i\in\Lambda} a_i}
=\frac{\langle\vect{a},\Gamma \vect{a}\rangle}{\langle\vect{a},{\bf 1}\rangle}.
\end{equation}
\end{corollary}

\begin{remark}
 Let $c$ be a constant, and let $\vect{a}$ (if it exists) be such that $\Gamma \vect{a} = c {\bf 1}$ and ${\bf a}\geq 0$. Then ${\bf a}$ is a critical point and 
$c$ is given by (\ref{C}).
\end{remark}

\begin{example}[The one-dimensional cycle]\label{CriticalCircle}
Assume that the $L$ cells are arranged on a cycle. The pure transport process ($D=0$) is reversible, so that
the critical points $\vect{a}>0$  of dynamical system (\ref{eq:jonsson}) are solutions
of linear system (\ref{SolAdjacency}). We illustrate some results given in Section \ref{Circle}.
When $L> 4$ is a multiple of 4, the set of critical points  $\vect{a}\in \R^L$ forms a two dimensional sub-manifold $M_c$ of $\R^L$ given by, when $\rho = 1/L$,
$$
M_c =\{(a_1,a_2,-a_1+2\rho,-a_2+2\rho,a_1,a_2,-a_1+2\rho,-a_2+2\rho,\cdots );\ a_k \in (0,2\rho),\ k=1, 2\}.
$$
When $L>4 $ is not a multiple of 4, $M_c$ is reduced to the uniform configuration $M_c =\{(\rho,\rho,\cdots,\rho)\}$. We will see that the uniform configuration is always unstable, and that the other critical points are unstable when  $\vect{a}>0$. However, the boundary points  are all  stable.
\end{example}

\subsubsection{General transport processes}

\begin{lemma}\label{CharacterizationIrreducible}
Assume that $G$ is connected, and that both $D$ and $T$ are positive. For $\vect{a}>0$, $f(\vect{a})=0$ if and only if there exists a constant $c$ such that
$\vect{a}$ solves the following system of quadratic equations:
\begin{equation}\label{QuadraticSystem}
(a_i -\frac{D}{T})N_i(\vect{a})+ a_i = c\ a_i N_i(\vect{a}),\ i=1,\cdots,\ \Lambda.
\end{equation}
\end{lemma}
\begin{proof}
Let $\mu_i = (a_i -D/T)N_i$, $i=1,\cdots,\Lambda$. Then $\vect{\mu} = (\mu_i)_{1\le i\le\Lambda}$ behave
\begin{eqnarray*}
(\mu Q(0,T,\vect{a}))_i &=&\sum_{j\sim i}\mu_j q_{ji}(\vect{a})+\mu_i q_{ii}(\vect{a})\\
                       &=&T \sum_{j\sim i} (a_j -\frac{D}{T})N_j \frac{a_i}{N_j}-T(a_i-\frac{D}{T})N_i \sum_{j\sim i}\frac{a_j}{N_i}\\
                       &=&T \sum_{j\sim i}(a_j -\frac{D}{T})a_i - T (a_i-\frac{D}{T})\sum_{j\sim i}a_j\\
                       &=& T \frac{D}{T}\sum_{j\sim i}(a_j -a_i),
 \end{eqnarray*}
 which gives the diffusion term contained in $f$. Hence, one can rewrite the equation $f(\vect{a})=0$ as
 $$(\vect{\mu} + \vect{a})Q(0,T,\vect{a})=0.$$
 By assumption, $\vect{a} > 0$ so that $Q(0,T,\vect{a})$ is irreducible as a Markov generator, and hence
 has only one invariant probability measure. The linear space composed of invariant measures is one-dimensional,
 so that the measure $\vect{\mu} + \vect{a}$ is proportional to $\pi(\vect{a})$. The result is a consequence of expression
 for $\pi(\vect{a})$ given in 
 (\ref{InvariantMeasureTransport}).
 \end{proof}

 The next paragraph generalizes the diffusive part to model the effect of potentials on the auxin flux.
 \subsubsection{Inclusion of potentials\label{Mechanical}}
 
 As stated in the Introduction, experiments have shown that both mechanical and biochemical processes play a role
 in plant patterning. We here adapt some ideas of \citet{Newell} and \citet{Newell2} to our
 discrete setting. The former considered the discrete model (\ref{eq:jonsson}) by taking a continuous limit,
  resulting in a   p.d.e. describing the time evolution of auxin concentrations, which is coupled to the von Karman equations from elasticity theory. These equations describe the deformations of an elastic shell or plate subject to various loading conditions. Usually, the in-plane stress is described using Airy  functions which are potential for the stress field. Here, we will simply suppose that this potential is given by some
 function $(\phi_i)_{1\le i\le L}$. We also suppose that the auxin flux is directed in part by these potentials and  assume a model of the form
 \begin{equation}\label{potential}
 \frac{{\rm d}a_i(t)}{{\rm d}t}=f_i(\vect{a})+\sum_{j\sim i}(a_j \phi_i -a_i \phi_j),
 \end{equation}
 $i=1,\ldots,L$. 
 We will see in the sequel that the critical points associated to (\ref{eq:jonsson}) exhibit regular geometrical patterns locally, but not necessarily globally. 
  The potential might be defined in such a way to reproduce the patterns obtained when considering mechanical buckling, and
  the model defined by (\ref{potential}) might then lead to more regularly spaced auxin peaks, see Figure \ref{Fig4}.
 
 \begin{lemma}\label{Coupling}
 Assume a model of the form (\ref{potential}), with $D>0$ and $T>0$. Let $\vect{a}>0$. Then
$ f_i(\vect{a})+\sum_{j\sim i}(a_j \phi_i -a_i \phi_j)=0$ if and only if there exists a constant
$c\in\R$ such that
$$(a_i-\frac{D}{T}-\frac{1}{T}\phi_i)N_i(\vect{a}) + a_i= c\ a_i N_i(\vect{a}),\ i=1,\cdots,\Lambda.$$
\end{lemma}
 
The proof of Lemma \ref{Coupling} is identical to the proof of Lemma \ref{CharacterizationIrreducible}.

\subsection{The reducible case \label{Reducible}}
We can adapt the previous notions to the case $D=0$ and reducible transition kernel $Q(0,T,\vect{a})$, that is when some $a_i$
vanish. In this case, there is a pair of nodes $i$ and $j$ such that
$$\prod_{k=1}^m q_{i_{k-1}i_k}(\vect{a})=0,$$
for all paths $\gamma:\ i_0=i\to i_1\to\cdots\to i_m =j$ taking
$i$ to $j$ in the graph $G=(\Lambda,E)$.

Example \ref{CriticalCircle} shows that the critical points associated with (\ref{eq:jonsson}) on a circle form a manifold when $L$
is a multiple of 4. We also assert that the boundary points obtained from $M_c$ by setting  $a_1=0$ are stable. We will thus consider subsets 
$I\subset \{1,\cdots,L\}$ corresponding to the sites $i$ where $a_i =0$.
We will denote by $\vect{a}\vert_I$ the restriction of any $\vect{a}$ to $I$. The same
notations apply for generators and adjacency matrices, where
one conserves only the transitions rates $q_{ij}(\vect{a})$
such that $i$, $j \in\Lambda\setminus I$.
According to Lemma \ref{ReversibilityPureTransport}, these sub-transition kernels
are reversible for $\vect{a}$ such that $a\vert_{\Lambda\setminus I} > 0$.
If one removes the nodes $i\in I$, the graphs decomposes as a product
of connected components $\gamma$, which 
form the sub-graph of $G$ induced by the nodes of $J=\Lambda\setminus I$.
The special form of the vector field associated
with (\ref{eq:jonsson}) ensures however that the set of critical values such
that $\vect{a}\vert_I =0$, $I\subset \{1,\cdots,L\}$, can be obtained by considering 
a family of transitions functions $Q_\gamma(0,T,\vect{a}\vert_\gamma)$.
For each component $\gamma$, Corollary \ref{CriticalAdjacency} shows that 
the related critical points are obtained by solving linear systems of the form
\begin{equation}\label{LocalSolution}
\Gamma_\gamma \vect{a}\vert_\gamma = c_\gamma {\rm\bf 1}\vert_\gamma,
\end{equation}
where $\Gamma_\gamma$ is the adjacency matrix of the sub-graph $\gamma$, and the $c_\gamma$ are normalization constants chosen in such a way
that $\sum_i a_i = \rho L$. 
The set of critical points is then obtained by
taking the direct product of the sets of critical values associated with
the sub-graphs $\gamma$.

\section{Asymptotic properties of the auxin flux for pure transport processes\label{s.Convergence} }
We consider the convergence of the dynamical system (\ref{eq:jonsson})
when $D=0$  using the method of Lyapunov functions. Suppose without loss
of generality that $T=1$.
We look for a function
$H(\vect{a})$ such that 
$$\frac{{\rm d}H(\vect{a}(t))}{{\rm d}t} =\langle\nabla H(\vect{a}(t)),\frac{{\rm d}\vect{a}(t)}{{\rm d}t}\rangle \le 0,\ \forall t \geq 0.$$
If furthermore this function is bounded, then $H(\vect{a}(t))$ converges, and we can in this way
get useful information concerning the convergence (e.g. toward the set of critical points) of
$\vect{a}(t)$ solution of (\ref{eq:jonsson}).

\begin{lemma}\label{Lyapunov}
Assume that $D=0$ and set $T=1$.
Let
\begin{equation}\label{LyapounovFunction}
H(\vect{a}) =-\frac{1}{2}\sum_{k\in\Lambda}a_k (N_k(\vect{a})+\kappa)=-\kappa \sum_{k\in\Lambda}a_k - \frac{1}{2} \sum_k\sum_{j\sim k}a_j a_k,
\end{equation}
where the functions $N_k(\vect{a})$ have been defined in (\ref{N}).
Let ${\vect a}(t)$ be a solution of the
o.d.e. (\ref{eq:jonsson}) such that   $a_i(0)\ge 0$.
Then
\begin{equation}\label{LyapounovFormula}
\frac{{\rm d}H(\vect{a}(t))}{{\rm d}t} =-\frac{1}{2}\sum_{k \in \Lambda}\sum_{j\sim k}q_{kj}q_{jk}(N_k-N_j)^2 \le  0, \forall t\geq 0.
\end{equation}

\end{lemma}
Notice that
\begin{equation*}
\frac{\partial H}{\partial a_k}(\vect{a})=- N_k(\vect{a}). 
 \end{equation*}
since the function $N_k =N_k(\vect{a})=\kappa+\sum_{j\sim k}a_j$ does not depend on the variable $a_k$.

\begin{proof}
One can write
\begin{eqnarray*}
\frac{{\rm d}H(\vect{a}(t))}{{\rm d}t} &=&-\sum_{k\in\Lambda}N_k \sum_{j\sim k} (a_j \frac{a_k}{N_j}-a_k\frac{a_j}{N_k})
        =- \sum_{k\in\Lambda}N_k\sum_{j\sim k}\frac{a_j}{N_k}\frac{a_k}{N_j}(N_k-N_j)\\
        &= & -\sum_{k\in\Lambda}N_k\sum_{j\sim k}q_{kj}q_{jk}(N_k-N_j) \\
         &=&-\frac{1}{2}\sum_{k \in \Lambda}\sum_{j\sim k}q_{kj}q_{jk}\Big(N_k(N_k-N_j)+N_j(N_j-N_k)\Big)\\
        &=&-\frac{1}{2}\sum_{k \in \Lambda}\sum_{j\sim k}q_{kj}q_{jk}(N_k-N_j)^2.
 \end{eqnarray*}
By Proposition \ref{Invariance}, $a_i(0)\ge 0$, $\forall i$, implies that $a_i(t)\ge 0$, $\forall i$, $\forall t > 0$,
so that $q_{kj}\ge 0$ and $q_{jk}\ge 0$, $\forall k\sim j$, and $\forall t > 0$, proving the  assertion.
 \end{proof}

To prove the convergence of the auxin flux, we use a Theorem   of Lyapunov- LaSalle (see \citet{LaSalle}).
Introduce the notation

$$\dot{H}(\vect{x})= \sum_{i=1}^L\frac{\partial H}{\partial x_i}f_i(\vect{x})=-\frac{1}{2}\sum_{k \in \Lambda}\sum_{j\sim k}q_{kj}q_{jk}(N_k-N_j)^2.$$
Consider the sets
$$
\Omega = \{\vect{x}\in[0, 2\rho]^L \, \mid \, \sum_i x_i=\rho L\}
\text{ and }
E_{\Omega} =\{\vect{x}\in\Omega \, \mid \, \dot{H}(\vect{x})=0\}.$$
\begin{lemma}\label{CriticalInvariant}
The set $E_\Omega$ is the set of critical points.
\end{lemma}
\begin{proof}
Let $x\in\Omega$. Then
 $\dot{H}(\vect{x})=0$ if and only if for all pairs $j\sim k$, either $x_j =0$, $x_k =0$ or
 $N_j = N_k$. Let $I_x :=\{i\in\Lambda;\ x_i =0\}$. Then $\dot{H}(\vect{x})=0$
 if and only if, for all pairs of neighbours $j\sim k$ such that $j\in\Lambda\setminus I_x$
 and $k\in\Lambda\setminus I_x$, one has that $N_j =N_k $. Let $\gamma$ be the connected
 component of the graph containing this pair (see Section \ref{Reducible}), with
 $N_j = N_k = c_\gamma$, for some positive constant $c_\gamma$. Then, $N_i\equiv c_\gamma$,
 $\forall i\in \gamma$. One then gets that $\dot{H}(\vect{x})=0$ if and only if
 the function $N$ is constant on the connected components $\gamma$ associated with $I_x$.
  Hence, for each such component, one
 has that $\Gamma_\gamma x\vert_\gamma = c_\gamma {\bf 1}\vert_\gamma$. The results is a consequence
 of Corollary \ref{CriticalAdjacency} and of the results of Section \ref{Reducible}.
\end{proof}

Let $M_{\Omega}$ be the largest invariant subset of $E_{\Omega}$. As $E_{\Omega}$ contains only the critical points of $f$, $E_{\Omega}$ is invariant. Hence, 
$M_{\Omega}=E_{\Omega}$.
\begin{proposition}\label{ConvergenceAuxin}
Let $\vect{a}(t)$ be the unique solution of the o.d.e. (\ref{eq:jonsson}) with
$\vect{a}(0)\in\Omega$. Then $\vect{a}(t)\in\Omega$, $\forall t > 0$ and $\vect{a}(t)$ converges
to $M_\Omega$ as $t\to\infty$.
\end{proposition}

\begin{proof}
Proposition \ref{Invariance} shows that the compact set  $\Omega$ is invariant.
 The continuously differentiable function $H$ is such that 
$\dot{H}(\vect{x})\leq 0$, $\forall x \in\Omega$. The results then follows from
a result of \citet{LaSalle}.
\end{proof}

\begin{corollary}\label{limit point}
 Every limit point of a trajectory $\bf{a}(t)$ is a critical point i.e. if for $t_n \nearrow \infty$, $a(t_n)\rightarrow a_{\infty}$ then $a_{\infty}\in M_{\Omega}$.
\end{corollary}

\begin{proof}
 If $a_{\infty}\not\in \Omega$, as $E_{\Omega}=M_{\Omega}$ is a closed set then $d()>0$. It's a contradiction with the proposition \ref{ConvergenceAuxin}.
\end{proof}

\begin{remark}[Global minimizers of $H$]\label{Global}
The literature contains results on
the set $\mu(G)$ of minimizers of $H$
when $\sum_{i\in \Lambda}a_i = 1$.
 The authors of \citep{Motzkin}
proved that $\max_{\vect{a}}\langle\vect{a},\Gamma \vect{a}\rangle= (\omega(G)-1)/\omega(G)$, where
$\omega(G)$ is the clique number of $G$, that is the order of the largest
complete sub-graph of $G$. Moreover, they obtained that the absolute minimum 
of $H$ is achieved at an interior point of the unit simplex if and only if $G$
is a complete multipartite graph. Various results were then obtained in \citep{Waller}.
where for example it is proved that $\mu(G)$ is a simplicial complex, having an automorphism
group similar to that of $G$. In some sense, $\mu(G)$ mirrors some of the geometry
of the graph $G$. 
\end{remark}

\begin{proposition}
If $D=0$, then system \eqref{Diff-ai} does not admit non-constant periodic solutions.
\end{proposition}

\begin{proof}
Every point of a periodic solution is a limit point and, according to our preceding results (corollary \ref{limit point}), it is a critical point. Unicity of a solution provides a contradiction.
\end{proof}

\begin{proposition}
If $D=0$, then the set of critical points of system \eqref{Diff-ai} is non-countable.
\end{proposition}

\begin{proof}
Let $\sum_{k=1}^L a_k(0)=C>0$. We know that the corresponding solution has to remain in the hyperplane $(\Pi): \sum_{k=1}^Lx_k=C$. Since the path is bounded it admits at least one limit point and, according to our preceding results (corollary \ref{limit point}), the latter is a critical point belonging to $(\Pi)$. Consequently, for every positive value of $C$, we obtain distinct critical points.
\end{proof}

\section{Stability of pure transport processes\label{SpecialClass}}

\subsection{The irreducible case}

We consider pure transport processes (i.e. $D=0$) on general graphs. We first discuss the
stability of the special class of critical points $\vect{a}> 0$ solving  equations of the form $\Gamma \vect{a} = c {\bf 1}$.
Without loss of generality, we set $T=1$.
 For such $\vect{a}$,
 $N_i(\vect{a})\equiv N=\kappa +c$, and therefore, when the graph is
regular, one obtains for example the uniform solution $\vect{a}=(\rho)=(\rho,\ldots,\rho)$.
When $G$ is the complete graph  $K_L$ of $L$ nodes, where every pair of nodes
$i\ne j$ are nearest neighbours, a simple computation shows that the Jacobian ${\rm d}f((\rho))$ associated with
(\ref{eq:jonsson}) and evaluated at the uniform configuration $(\rho)$, is given by
$$\frac{\partial f_i ((\rho))}{\partial a_j}=\frac{\rho^2}{N^2},
\ \frac{\partial f_i ((\rho))}{\partial a_i}=-\sum_{j\ne i}\frac{\rho^2}{N^2}.$$
Consequently, ${\rm d}f((\rho))$ is a symmetric generator, and thus admits only non-positive  real eigenvalues. The uniform configuration is then stable for the complete graph.

 \begin{proposition}\label{StabilityGeneral}
 Let $\vect{a} > 0$ be such that $\Gamma \vect{a} = c {\bf 1}$, for some positive constant $c>0$. According to Lemma \ref{CriticalAdjacency},
 $\vect{a}$ is a critical point, with $N_i(\vect{a})\equiv N = c+\kappa$.
 Assume that $D=0$ and set $T=1$.
  The Jacobian
 ${\rm d}f(\vect{a}) = (\partial f_i/\partial a_j)$ evaluated at $\vect{a}$ is then given by
 $${\rm d}f(\vect{a})=\frac{1}{N^2}d(\vect{a})\Gamma \Big(c\ {\rm id}-d(\vect{a})\Gamma\Big),$$
 where $d(\vect{a})$ is the diagonal matrix of diagonal given by
 $\vect{a}$, and where $\Gamma$ is the adjacency of the graph. 
 \end{proposition}
 
 The proof of Proposition \ref{StabilityGeneral} is given in  Section \ref{Appendix}.

We  now characterize the set of stable configurations using the spectral gap of
the matrix $P(\vect{a})=\Gamma d(\vect{a})/c$. Let $P$ be a stochastic matrix associated with a Markov
chain on the state space $\Lambda$. We assume that $P$ is reversible with invariant probability
measure $\pi$. Let $A$ be the matrix defined by
$A_{ij}=\pi_i p_{ij}\equiv \pi_j p_{ji}$, $i\ne j$. The eigenvalues of $P$ are real, given by
$-1\le\beta_L\le\cdots\beta_2 <\beta_1 = 1$, and the spectral gap is given by $C=1-\beta_2$.
Let $L = {\rm id}-P$ be the associated Laplace operator, of eigenvalues $\lambda_k = 1-\beta_k$,
$k=1,\cdots, L$. Then (see e.g. \citep{Diaconis})
\begin{equation}\label{Gap}
C =\lambda_2 = \inf\{\frac{{\cal E}_\pi(\phi,\phi)}{{\rm Var}_\pi(\phi)}:\ \phi \text{ is nonconstant}\},
\end{equation}
where 
$${\cal E}_\pi (\phi,\phi)=\frac{1}{2}\sum_{i,j}(\phi(j)-\phi(i))^2 A_{ij},$$
is the Dirichlet form associated with $L$, and where ${\rm Var}_\pi (\phi)$
is the variance of the random variable $\phi$ with respect to the invariant probability measure $\pi$.
One can check that
$${\rm Var}_\pi (\phi)=\frac{1}{2}\sum_{i,j}(\phi(j)-\phi(i))^2 \pi_i \pi_j.$$
We can also reformulate the above variational problem in a different way: set
$\langle \phi\rangle_\pi = \sum_{i\in\Lambda}\phi (i)\pi_i$. Then
\begin{equation}\label{Gap2}
C = \inf\{\frac{{\cal E}_\pi(\phi,\phi)}{{\rm Var}_\pi(\phi)}:\ \langle\phi\rangle_\pi =0 \}.
\end{equation}

 \begin{lemma}\label{SpectralGeneral}
 Let $G$ be a connected graph of adjacency matrix $\Gamma$, and let 
 $\vect{a}>0$ satisfy  $\Gamma \vect{a} = c\ {\bf 1}$ for some $c>0$. The matrix $P(\vect{a})$ defined by
 \begin{equation}\label{StochasticMatrix}
 P(\vect{a}) = \frac{1}{c}\Gamma d(\vect{a}),
 \end{equation}
 is stochastic, irreducible, reversible, of invariant measure $\pi'(\vect{a})$ given by
 $\pi'(\vect{a})_i = a_i/(\rho L)$, and with a real spectrum
 $-1\le \beta_{\Lambda}\le \beta_{\Lambda-1}\le \cdots\le \beta_2< \beta_1 = 1$. Let $ C(\vect{a})$ be the 
 spectral gap of $P(\vect{a})$, defined by $C(\vect{a}) = 1-\beta_2$.
  $\vect{a}$ is stable if and only if
 $C(\vect{a}) \ge 1$. Moreover, the spectral gap is given by
 $$C(\vect{a})=\delta \inf_\phi  
 \frac{\sum_{i,j}(\phi(j)-\phi(i))^2 \gamma_{ij} \pi'(\vect{a})_i   \pi'(\vect{a})_j}{
   \sum_{i,j}(\phi(j)-\phi(i))^2  \pi'(\vect{a})_i   \pi'(\vect{a})_j}\le \delta,
   $$
   where $\delta = \rho L /c > 1$, and where
    the infimum is taken over all nonconstant functions $\phi$.
 \end{lemma}
 
\begin{proof}
The matrix is stochastic since by assumption
$\Gamma \vect{a}= c \vect{1}$.

  Let $\pi'(\vect{a})=\left(\frac{a_i}{\rho L}\right)_{i\in \Lambda}$. Then
  $P(\vect{a})$ is reversible of invariant measure given by $\pi'$.
Notice next that
$$ A_{ij}=\pi'(\vect{a})_i P(\vect{a})_{ij}=\delta \gamma_{ij}\pi'(\vect{a})_i   \pi'(\vect{a})_j,$$
where we recall that $\gamma_{ij}\in\{0,1\}$ is the $(i,j)$ entry of the adjacency matrix $\Gamma$.
Hence, using the variational characterization of the spectral gap given in (\ref{Gap}), 
$$
C \le   \frac{{\cal E}_{\pi'(\vect{a})}(\phi,\phi)}{{\rm Var}_{\pi'(\vect{a})}(\phi)}
 = \delta \frac{\sum_{i,j}(\phi(j)-\phi(i))^2 \gamma_{ij} \pi'(\vect{a})_i   \pi'(\vect{a})_j}{
   \sum_{i,j}(\phi(j)-\phi(i))^2  \pi'(\vect{a})_i   \pi'(\vect{a})_j} \le \delta,
 $$
when $\phi$ is  non-constant.
The configuration $\vect{a}$ is stable if and only if the eigenvalues of the Jacobian matrix $df(\vect{a})$ given in the proposition \ref{StabilityGeneral}  are all non-positive. 
The adjacency matrix $\Gamma$ is symmetric, so that
 $\left(\Gamma d(\vect{a})\right)^* = d(\vect{a})\Gamma$. It follows that
  the eigenvalues $\tilde{\beta}_i$ of $d(\vect{a})\Gamma$ are equal to $c \beta_i$,
$i=1,\cdots,L$.  
   The eigenvalues of $N^2 df(\vect{a})$ are given by $ \tilde{\beta}_i (c-\tilde{\beta}_i )= \beta_i (1-\beta_i ) c$. Hence,   $\vect{a}$ is stable if and only if $\beta_2 < 0$, that is  if and only if  $C \geq 1$.
  \end{proof}

   \begin{corollary}\label{Unstable}
Let $G$ be a connected graph of adjacency matrix $\Gamma$, and let $\vect{a}>0$ satisfy  $\Gamma \vect{a} = c\ {\bf 1}$ for some $c>0$.
 For $i\in\Lambda$, let ${\cal V}_i = \{j\in\Lambda;\ j\sim i\}$
 be the neighbourhood of $i$.
  Assume that there exist elements
 $i_0$, $i_1$, $i_2$ and $i_3$ of $\Lambda$ such that
 \begin{equation}\label{InstabilityCriterium}
 i_1\in {\cal V}_{i_0}, \ i_2 \in {\cal V}_{i_1}\setminus {\cal V}_{i_0}\setminus \{i_0\},\ 
 i_3 \in {\cal V}_{i_2}\setminus {\cal V}_{i_0}\setminus \{i_0\}.
 \end{equation}
 Then $\vect{a}$ is unstable.
 \end{corollary}

 \begin{example}\label{ExampleSubgaphs}
 When $G$ is a sub-graph of a two-dimensional grid, 
 a solution to the linear system $\Gamma \vect{a} = c {\bf 1}$ can possibly
 to be stable only when $G$ belongs to the list given in Figure \ref{subgraph},
 which consists in the square, the star, and all the various parts of the star. 
 \end{example}

 \begin{proof}
 We use Lemma \ref{SpectralGeneral} to express the spectral gap of $P(\vect{a})$
 as
 \begin{eqnarray*}
 C(\vect{a})&=&\delta \inf_{\langle\phi\rangle_{\pi'(\vect{a})}=0}
 \frac{\sum_i \phi (i)^2 \pi'(\vect{a})_i \sum_j \gamma_{ij}\frac{a_j}{\rho L} -\sum_{i,j}\gamma_{ij}\phi (i) \phi (j)
 \pi'(\vect{a})_i \pi'(\vect{a})_j}{\sum_i \phi (i)^2 \pi'(\vect{a})_i}\\
 &=&\delta \inf_{\langle\phi\rangle_{\pi'(\vect{a})}=0}
 \frac{\sum_i \phi(i)^2 \pi'(\vect{a})_i \frac{c}{\rho L} - \sum_{i,j}\gamma_{ij}\phi(i)\phi(j)\pi'(\vect{a})_i \pi'(\vect{a})_j}
 { \sum_i \phi(i)^2 \pi'(\vect{a})_i }\\
&=&
 \delta  \inf_{\langle\phi\rangle_{\pi'(\vect{a})}=0}
 \frac{\sum_i \phi(i)^2 \pi'(\vect{a})_i \delta^{-1} - \sum_{i,j}\gamma_{ij}\phi(i)\phi(j)\pi'(\vect{a})_i \pi'(\vect{a})_j}
 { \sum_i \phi(i)^2 \pi'(\vect{a})_i }
 \end{eqnarray*}
 We will prove that $C(\vect{a}) < 1$ by choosing a test function $\phi$
 satisfying ${\langle\phi\rangle_{\pi'(\vect{a})}=0}$ for which
 $$ \delta \frac{\sum_i \phi(i)^2 \pi'(\vect{a})_i  \delta^{-1} - \sum_{i,j}\gamma_{ij}\phi(i)\phi(j)\pi'(\vect{a})_i \pi'(\vect{a})_j}
 { \sum_i \phi(i)^2 \pi'(\vect{a})_i }
  < 1,$$
which is equivalent to require that
 $$\sum_{i,j}\gamma_{ij}\phi(i)\phi(j)\pi'(\vect{a})_i \pi'(\vect{a})_j > 0.$$
We set $\phi(j) = 0$, $\forall j\in {\cal V}_{i_0}$.
 For  $j\in \Lambda\setminus {\cal V}_{i_0}\setminus \{i_0\}$,
 we choose $\phi(j)$ to be arbitrary but positive. For $j=i_0$,
 we choose $\phi(i_0)$ so that
 $$a_{i_0}\phi (i_0)=-\sum_{j\ne i_0}a_j \phi(j).$$
 Consequently $\langle\phi\rangle_{\pi'(\vect{a})}=0$ and
 $\sum_{i,j}\gamma_{ij}\phi(i)\phi(j)a_i a_j > 0$.
 \end{proof}


   Corollary \ref{Unstable} provides a simple condition ensuring the non-stability of
configurations $\vect{a}$ satisfying $\Gamma \vect{a} = c {\bf 1}$. We next consider  
the reducible case where $a_i =0$ for $i\in I\subset \Lambda$. Set $J=\Lambda\setminus I$, and let
$\{\gamma_1,\cdots,\gamma_P\}$ be the collection of sub-graphs of $G$ 
induced by the nodes of $J$, of node set $J_{\gamma_p}$ and of
 adjacency matrices $\Gamma_{\gamma_p}$, $p=1,\cdots, P$. We again assume that
  $\Gamma_{\gamma_p}\vect{a}\vert_{\gamma_p}=c_{\gamma_p} \vect{1}$ for some $c_{\gamma_p}>0$.

   \subsection{The reducible case}
We consider the stability of critical points $\vect{a}$ such that 
$a_i = 0$, for $i\in I \subset\Lambda$ with $I\ne \emptyset$.

\begin{proposition}\label{StabilityBoundary}
Assume that $D=0$ and set $T=1$.
Let $\vect{a}$ be a critical point of \eqref{eq:jonsson} such that $a_i=0$ for $i \in I$.
Let $\{\gamma_1,\cdots,\gamma_P\}$ be the collection of sub-graphs of $G$ obtained
by deleting the nodes of $I$, of adjacency matrices $\Gamma_{\gamma_p}$, $p=1,\cdots, P$. 
The critical points $\vect{a}$ are obtained by solving linear systems of the form
  $\Gamma_{\gamma_p}\vect{a}\vert_{\gamma_p}=c_{\gamma_p} \vect{1}\vert_{\gamma_p}$ for some $c_{\gamma_p}>0$
  (see Section \ref{Reducible}). The spectrum of the Jacobian evaluated at $\vect{a}$ is given by
\begin{equation}
{\rm spec}({\rm d}f(a))= \bigcup_{p=1}^P {\rm spec}\left({\rm d}f\vert_{\gamma_p}(\vect{a}\vert_{\gamma_p})\right) \cup \left\{\sum_{k\sim i} \frac{a_k}{N_k} - \frac{N_i-\kappa}{N_i}, i\in I \right\}
\end{equation}
\end{proposition}

The proof of Proposition \ref{StabilityBoundary} is given in  Section \ref{Appendix}.

Proposition \ref{StabilityBoundary} shows that such configurations are stable when
1) each $a\vert_{\gamma_p}$ is stable and 2) when $\sum_{k\sim i} a_k / N_k(\vect{a})-(N_i(\vect{a})-\kappa)/N_i(\vect{a}) < 0$,
$i\in I$. Here, if $k\sim i$, $i\in I$, $k\in \Lambda\setminus I$, $N_k(a)$ is given by
the constant $\kappa + c_{\gamma_p}$ when $k\in J_{\gamma_p}$. To go further, we need the following

\begin{definition}
Let $J\subset\Lambda$. The outer boundary of $J$, denoted by $\partial J$, is the subset
of $\Lambda$ given by
$$\partial J = \{j \in\Lambda\setminus J;\ j\sim J\}.$$
\end{definition}

\subsection{Example: the rectangular grid\label{StabilityGrid}}

We now illustrate the various stable patches we can form by using the 
building blocks, as given in Figures \ref{subgraph} and \ref{stable}. It is easy to provide
examples of unstable configurations when the outer boundary of some component $\gamma$
is such that
\begin{equation}\label{NotIsolated}
\partial\Big(\partial J_\gamma\Big)\cap J_{\gamma'} \ne\emptyset, \text{ for some component } \gamma'\ne\gamma,
\end{equation}
as illustrated in  Figure \ref{FigureStable}(a).

\begin{figure}[h!]
\begin{minipage}{0.5\textwidth}
\begin{center}
\includegraphics[height=4cm]{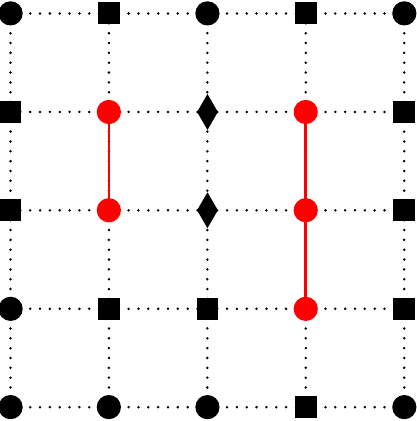}

\vspace{3mm}
(a) Unstable configuration
\end{center}
\end{minipage}
\begin{minipage}{0.5\textwidth}
\begin{center}
\includegraphics[height=4cm]{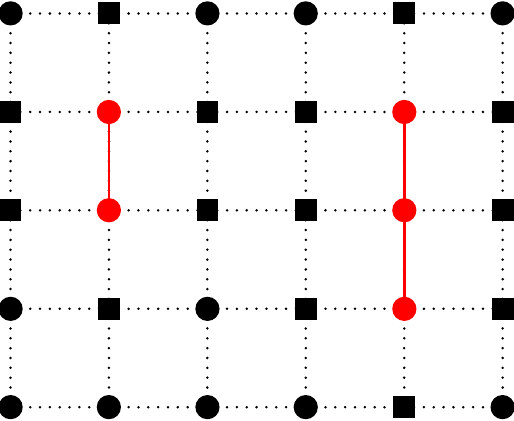}

\vspace{3mm}
(b) Stable configuration
\end{center}
\end{minipage}
\caption{(a) One can check in this example that (\ref{NotIsolated}) implies the non-stability of the configuration for well chosen parameters. Red dots indicates cells with $a_i \neq 0$. (b) One can check in this example that (\ref{Isolated}) is satisfied, ensuring the stability of the configuration. \label{FigureStable}}
\end{figure}

Next, the reader can verify, using Proposition \ref{StabilityBoundary}, that any patch composed of building blocks
disposed in such a way that 
\begin{equation}\label{Isolated}
\partial\Big(\partial J_p \cup J_p\Big)\cap \Big(\cup_{p'\ne p}J_{p'}\Big) = \emptyset,\  \forall p=1,\cdots, P,
\end{equation}
is stable. Figure \ref{FigureStable}(b) exhibits a typical example of  a stable configuration in this setting.

\subsection{Example: the pure transport process on the circle \label{Circle}}
We here assume that $D=0$ and $T=1$.
Corollary \ref{Unstable} yields  the instability of uniform solution $(\rho)=(\rho, \ldots, \rho)$
when the length $L$ of the cycle is larger than 4.
The adjacency matrix of the circle is circulant, with eigenvalues given by
$$\mu_k= e^{2\pi i \frac{k}{L}}+e^{2\pi i \frac{(L-1)k}{L}}=2\cos\left(2\pi \frac{k}{L}\right).$$
The determinant of $\Gamma$ vanishes if and only if there exists $j\in\{1,...,L\}$ such that $\mu_j=0$, that is if
$$\cos\left(2\pi \frac{j}{L}\right)=0 \Leftrightarrow 2\pi \frac{j}{L}= \frac{\pi}{2}+k\pi \text{ for } k\in\N,$$
or equivalently if there is a $k\in\N$ such that
$j=\frac{L}{4} +k\frac{L}{2}\in \N$.
Hence, the determinant of $\Gamma$ vanishes if and only if $L$ is a multiple of 4. 
In this case, the set  $M_c$  of critical values $\vect{a}$ (that is satisfying $\Gamma \vect{a}=c \vect{1}$) such that $a_i > 0$, $\forall i\in\Lambda$, is such that
$$ a_3=  c-a_1, a_4=  c-a_2, a_5=  a_1, a_6=  a_2, a_7=  c-a_1, ...$$
with $a_1\neq 0 \neq a_2$. 
Recalling that we impose the following normalization
$\sum_{i=1}^L a_i=\rho L$, we obtain 
$$\sum_{i=1}^L a_i=\rho L \Leftrightarrow 2c\frac{L}{4}= \rho L \Leftrightarrow c=2 \rho.$$
The set of critical values $M_c$ is then composed of configurations of the form
$$\vect{a}=(a_1,a_2,2 \rho-a_1,2 \rho-a_2 ,a_1, a_2,2 \rho-a_1 ,2 \rho-a_2 ,...,a_1,a_2,2 \rho-a_1 ,2 \rho-a_2)$$
with $(a_1,a_2)\in (0,2\rho)\times (0,2\rho)$. 
Corollary \ref{Unstable} then implies that this set contains only unstable points when $L>4$. For $L=4$, the critical point  $\vect{a}=(a_1,a_2,2 \rho-a_1,2 \rho-a_2)$ is stable since the eigenvalues of the Jacobian matrix are such that
$$\lambda_1=\lambda_2=\lambda_3=0 \text{ and } \lambda_4=-\frac{2c}{(\kappa+c)^2}$$
We can summarize these results in the following corollary:

\begin{corollary}\label{StabilityCriticalOneDim}
Assume that the nodes are arranged on a circle of size $L$. The set $M_c$ of critical values $\vect{a}>0$ such that $f(\vect{a})=0$ contains only the uniform configuration $(\rho,...,\rho)$ if L is not a multiple of 4. In the case where $L = 4n$, for some $n\in \N$ with $n\geq 1$, $M_c$ is given by
$$
M_c =\{(a_1,a_2,-a_1+2\rho,-a_2+2\rho,a_1,a_2,-a_1+2\rho,-a_2+2\rho,\cdots );
\ a_k \in (0,2\rho),\ k=1, 2\}.
$$
Any element of $M_c$ is unstable except for $L=4$. 
\end{corollary}
The set $M_c^{tot}$ of all critical points is obtained by decomposing the circle into sub-graph $\gamma$ such  $\vect{a}\vert_\gamma>0$ and by solving the system 
$$\Gamma_\gamma \vect{a}\vert_\gamma = c_\gamma {\rm\bf 1}\vert_\gamma,$$
for these sub-graphs. We can prove that this system has positive solution $\vect{a}\vert_\gamma$ if and only $|I|<4$ ($|I|:=$ length of the path), because for $|I|\geq 4$, we see that $a_4=0$ (which is in contradiction with the hypothesis).
When $|I|=3$, the critical points take the form
$\vect{a}\vert_\gamma=(z_1,c_\gamma,c_\gamma-z_1)$,
with $z_1\in(0,c_\gamma)$ and when $|I|=2$, $\vect{a}\vert_\gamma=(c_\gamma,c_\gamma)$. In these two cases, the critical points are stable as the Lyapunov function H defined in (\ref{LyapounovFunction}) takes its minimal value $H(a)=-(\kappa+\frac{\rho L}{4})\rho L$. The global minimum of H is obtained by adapting the result of \citet{Motzkin}, see Remark \ref{Global}. Finally, if $|I|=1$, we have
$\vect{a}\vert_\gamma=(c_\gamma)$; $H$ is maximal and hence $\vect{a}$ is unstable.

The set $M_c^{tot}$ of critical points is then obtained by taking the direct product of the sets of critical values associated with
the paths $\gamma$. 
For example, if L is a multiple of 4, the subset of $M_c^{tot}$ defined by
$$
\tilde M_c  =
\{(a_1,a_2,-a_1+2\rho,-a_2+2\rho,a_1,a_2,-a_1+2\rho,-a_2+2\rho,\cdots );
\ a_1=0,\ a_2 \in (0,2\rho)\},
$$
is composed of critical values which are stable  since
$$\lambda=0 \text{ with multiplicity } 3\frac{L}{4} \text{ and } \lambda=\frac{-2c^2}{(\kappa+c)^2} \text{ with multiplicity } \frac{L}{4} $$

\subsection{An explicit computation when $D=0$ on the circle\label{SpecialComputation} }

As we have seen, when $\vert I\vert =3$, the stable configurations are given by
 triplets of the form $(z_1,c_\gamma,c_\gamma -z_1)$, where $z_1$ is such that
$z_1 \in (0,c_\gamma)$, for some positive constant $c_\gamma > 0$. 

Consider a path composed of five cells  $i-1$, $i$, $i+1$, $i+2$ and $i+3$ such that
$a_{i-1}=a_{i+3}=0$, so that the dynamical system (\ref{eq:jonsson}) 
associated with these cells
becomes
\begin{eqnarray}
\frac{{\rm d}a_i}{{\rm d}t}&=& \frac{a_{i+1} a_i}{\kappa + a_{i}+a_{i+2}}- \frac{a_i a_{i+1}}{\kappa + a_{i+1}},\label{eq1}\\
\frac{{\rm d}a_{i+2}}{{\rm d}t}&=& \frac{a_{i+1} a_{i+2}}{\kappa + a_{i}+a_{i+2}}- \frac{a_{i+2} a_{i+1}}{\kappa + a_{i+1}},\label{eq2}\\
\frac{{\rm d}a_{i+1}}{{\rm d}t}&=&  \frac{a_i a_{i+1}}{\kappa + a_{i+1}}+\frac{a_{i+2} a_{i+1}}{\kappa + a_{i+1}}
-\frac{a_{i+1} a_i}{\kappa + a_{i}+a_{i+2}}-\frac{a_{i+1} a_{i+2}}{\kappa + a_{i}+a_{i+2}}.\label{eq3}
\end{eqnarray}
Dividing (\ref{eq1}) by (\ref{eq2}) yields that
$$\frac{\frac{{\rm d}a_i}{{\rm d}t}}{\frac{{\rm d}a_{i+2}}{{\rm d}t}}
=\frac{a_i}{a_{i+2}}.$$
Thus there is a positive constant $c>0$ such that
\begin{equation}\label{Relation1}
a_{i+2} = c a_i.
\end{equation}
Plugging this identity in (\ref{eq3}), one obtains 
$$\frac{{\rm d}a_{i+1}}{{\rm d}t} = (1+c)a_i a_{i+1}(\frac{1}{\kappa + a_{i+1}}-\frac{1}{\kappa + a_i +a_{i+2}}),$$
and finally
$$\frac{\frac{{\rm d}a_{i+1}}{{\rm d}t}}{\frac{{\rm d}a_{i}}{{\rm d}t}}
= -(1+c).$$
Hence there exists a constant $d$ such that $a_{i+1}=d-(1+c)a_i$.
Normalizing the total mass in such a way that
$a_i +a_{i+1}+a_{i+2} = 3\rho$,
one gets that $3\rho = d$ and
\begin{equation}\label{Relation2}
a_{i+1}= 3\rho -(1+c)a_i.
\end{equation}
Plugging (\ref{Relation1}) and (\ref{Relation2}) in equation (\ref{eq1}) yields the differential equation
$$
\frac{{\rm d}a_i}{{\rm d}t}
= \frac{ a_i (3\rho -(1+c)a_i)(3\rho -2(1+c)a_i)}{(\kappa + a_i(1+c))(3\rho +\kappa -(1+c)a_i)}.$$
Setting
$u = (1+c)a_i$, one gets the o.d.e.
$$\frac{{\rm d}u}{{\rm d}t} = \frac{u(3\rho -u)(3\rho -2u)}{(\kappa +u)(3\rho +\kappa -u)}.$$
Solving by partial fractions expansions, one obtains
$$\frac{3\kappa\rho +\kappa^2}{9\rho^2}(\ln(u)
+\ln(3\rho-u))-\frac{9 \rho^ 2 +4(\kappa^2+3 \rho \kappa)}{18\rho^2}\ln(3\rho-2u)
=t+\alpha,$$
for some constant $\alpha$.  Clearly one must have $u<3\rho /2$. 

\begin{lemma}
\label{lem}
As $t \to \infty$, $u(t)=(1+c)a_i(t) \longrightarrow \frac{3\rho}{2}$. 
\end{lemma}

\begin{proof}
The preceding considerations show that we have to consider only initial conditions of the form  $0\leq u(0)\leq 3\rho$. Clearly $0, \frac{3\rho}{2}$ and $3\rho$ are critical points of our equation.

 We can easily find a compact interval $I$ whose interior contains $J=[0,3\rho]$ and so that $f'(u)$ is continuous and thus bounded over $I$. As a consequence $f$ satisfies a Lipschitz-condition over $I$. According to the general theory, for any initial condition $u(0)\in J$ our equation admits a unique solution defined over a maximal interval $I_m$. If $u(0)=0$, then $u\equiv 0$ is the corresponding solution. If $u(0)\in ]0,\frac{3\rho}{2}[$, then $\dot u(0)>0$. Due to unicity, the solution can not reach a critical point  in a finite time and thus the boundary of $]0,\frac{3\rho}{2}[$. Moreover the solution is obviously bounded entailing $I_m=[0,+\infty[$. For the preceding reasons the derivative of $u(t)$ is never $0$ and thus always positive since $\dot u(0)>0$. Thus $u(t)$ increases to $\frac{3\rho}{2}$ as $t\rightarrow +\infty$. The same reasoning shows that $u(t)$ decreases to $\frac{3\rho}{2}$ as $t\rightarrow +\infty$ for $u(0)\in ]\frac{3\rho}{2}, 3\rho[$. Finally if $u(0)=3\rho$, then $u\equiv 3\rho$.
\end{proof}

Furthermore, (\ref{Relation1}) yields
$$c=\frac{a_{i+2}(0)}{a_i(0)}.$$
As $(1+c)a_i=a_i+a_{i+2}$ tends to $c_{\gamma}$ as time goes to infinity, Lemma \ref{lem} yields that
$c_{\gamma}=3\rho/2$,
and
$$a_i(t) \longrightarrow \frac{3\rho}{2(1+c)}=\frac{c_{\gamma}}{1+c},$$
as $t \to \infty$.
 \eqref{Relation1} and \eqref{Relation2} show that
$$a_{i+1}=3\rho-(1+c)a_i \longrightarrow \frac{3\rho}{2}= c_{\gamma} \text{ and } a_{i+2}=c a_i\longrightarrow \frac{c}{1+c}c_{\gamma}=c_{\gamma}-\frac{c_{\gamma}}{1+c}.$$
In summary, one obtains that an orbit defined by initial conditions of the form
$$(a_{i-1}(0),a_i(0),a_{i+1}(0),a_{i+2}(0),a_{i+3}(0)) \text{ with } a_{i-1}(0)=a_{i+3}(0)=0$$
converges to the critical point
$(z_1,c_\gamma,c_\gamma -z_1)$,
with $z_1= \frac{c_{\gamma}}{1+c}$, $c_\gamma=\frac{3\rho}{2}$ and ${c=\frac{a_{i+2}(0)}{a_i(0)}}$.
Finally, if the system starts from a symmetric initial state  $a_i(0)=a_{i+2}(0)$, the constant c is egal to 1 and the system tends to $(0,\frac{3\rho}{4},\frac{3\rho}{2},\frac{3\rho}{4},0)$ as $t \to \infty$.

    
\section{Appendix\label{Appendix}}
\subsection{ Proof of Theorem \ref{Invariance}}
First, we easily check that the system $\dot{\vect{a}}=f(\vect{a})$ is conservative, i.e. 
 $$\forall t \in \R_{\geq 0}, \; \sum_i^L a_i(t) =\sum_i^L a_i(0).$$
 In the following, we use the notation $\dot{\vect{a}}$ instead of $\frac{{\rm d}\vect{a}}{{\rm d}t}$.
The latter is equivalent to 
$$\sum_i^L \dot{a}_i(t) = \sum_i^L f_i(\vect{a})= 0. $$
In fact, one can write
\begin{eqnarray*}
\sum_i \dot{a}_i (t)&=&D \sum_i\sum_{k \sim i} (a_k -a_i)+T \sum_i \sum_{k\sim i}a_k a_i\left(\frac{N_i-N_k}{N_k N_i}\right)\\
&=&D \sum_i(d_i a_i-d_i a_i)+2 T \sum_{k\sim i}\frac{a_k a_i}{N_k N_i}((N_i-N_k)-(N_k-N_i))=0,
\end{eqnarray*}
where $N_k=\kappa+\sum_{j\sim k}a_k$, and where $d_i$ is the degree of i (that is the number of neighbours of i). 

Next,  system (\ref{eq:jonsson}) can be written as 
\begin{equation}\label{eq:jonssonb}
\dot a_i=D\sum_{k\sim i}a_k+T\sum_{k\sim i}\left(\frac{a_k}{\kappa+\sum_{j\sim k}a_j}-\frac{a_k}{\kappa+\sum_{j\sim i}a_j}-\frac{D}{T}\right)a_i. 
\end{equation}
Let $\vect{a}$ a solution of (\ref{eq:jonssonb}) with $\vect{a}(0)\in\R^L_{\geq 0}$.

We say that the function $f:\R_+\rightarrow \R$ is instantaneously positive (i.p.) if there exists $\delta>0$ so that $f$ is strictly positive over $(0,\delta)$. If $f(0)>0$ and $f$ is continuous to the right at $0$, then $f$ is i.p.. It is also clear that if $f$ admits a 
strictly positive right-hand derivative at $0$, then it is i.p.. 

Let $U$ be the open set $U=\{\vect{x}=(x_1,x_2,...,x_L)\in\R^L;-\frac{\kappa}{2L}<x_i\}$. Since the right-hand member of (\ref{eq:jonssonb}) is continous over $U$, the general theory of o.d.e.'s provides the existence of a solution defined over a maximal interval $0\in J^+\subset \R_+$ for any initial condition $\vect{ a}(0)\in U$. Moreover, the solution is unique because the right-hand member of (\ref{eq:jonssonb}) locally lipschitzian. Set for convenience 
 \begin{eqnarray*}
 h_i(t)&=& D\sum_{k\sim i}a_k(t)\quad \hbox{ and } \\
 g_i(t)&=&T\sum_{k\sim i}\left(\frac{a_k(t)}{\kappa+\sum_{j\sim k}a_j(t)}-\frac{a_k(t)}{\kappa+\sum_{j\sim i}a_j(t)}-\frac{D}{T}\right).
 \end{eqnarray*}
The variation of constants formula allows us to write , $\forall t\in J^+$,
\begin{equation}\label{F}
a_i(t)=a_i(0)e^{\int_0^tg_i(s)ds}+\int_0^t h_i(u)e^{-\int_u^tg_i(v)dv}du.
\end{equation}
Since $a_i(0)\geq 0$, the first term in (\ref{F}) is non-negative. Moreover if $a_k(t)$ is i.p.
for some $k\sim i$, then according to (\ref{F}), the same property holds for $a_i(t)$. In particular, if $a_k(0)>0$ for some $k\sim i$, then by continuity $a_k(t)$ is i.p. and thus also $a_i(t)$.

\noindent \textbf{The case $D>0$:}

Clearly, if $\vect{ a}(0)= \vect{0}$, then the unique solution is identically $0$. Otherwise, there exists $1\leq i_0\leq L$ with $a_{i_0}(0)>0$ and $\forall j\sim i_0, a_j(t)$ is i.p.. Since our graph is supposed to be connected, every $i$ admits a neighbor $k\sim i$ with $a_k(t)$ i.p..   Hence, $a_i(t)$ is i.p. $\forall i,1\leq i\leq L$.

The preceding arguments show that for any initial condition $\vect{ a}(0)\in \R_{\geq 0}^L\subset U$, all components of the solution of (\ref{eq:jonssonb}) are i.p.. Let us suppose that one of them admits the value $0$ in $J^+\backslash\{0\}$. Since all components are continuous and their number is finite, there exists a first time $t_0>0$ for which at least one component $a_{i_0}(t_0)=0$ and all of them are strictly positive over $(0,t_0)$. According to (\ref{F}),  we have
$$a_{i_0}(t_0)=0=a_i(0)e^{\int_0^{t_0}g_i(s)ds}+\int_0^{t_0}h_i(u)e^{-\int_u^{t_0}g_i(v)dv}du.$$
Clearly $h_i(t)>0$ over $J^+\backslash\{0\}$ and since the first term is non-negative, we conclude to $a_{i_0}(t_0)>0$, a contradiction. Therefore all $a_i(t)$ are strictly positive over $J^+\backslash \{0\}$.

\noindent \textbf{The case $D=0$:}

If $a_i(0)=0$, the homogeneous equation for $a_i(t)$ admits only the zero solution, and we remove the related $i$th component from (\ref{eq:jonssonb}). Otherwise $a_i(0)>0$ and, by continuity, $a_i(t)$ is i.p.. In that case $a_i(t)=a_i(0)e^{\int_0^tg_i(s)ds}>0$ over $J^+$.\\
\\
\\
In both cases the solution of (\ref{eq:jonssonb}) have strictly positive components over $J^+$. We also proved that $\forall t \in J^+$ we have:
$$\sum_{1\leq i\leq L}a_i(t)=\sum_{1\leq i\leq L}a_i(0).$$
As a consequence the solution of (\ref{eq:jonssonb}) is bounded and thus the unique solution of our problem is defined over $J^+=[0,+\infty)$.

\subsection{Proof of Proposition \ref{StabilityGeneral}}
 We first give the Jacobian, for general $\vect{a}$. We have
\begin{equation}\label{Jacobian1}
\frac{\partial f_i(\vect{a})}{\partial a_j}
=\frac{a_i}{N_j}-\frac{a_i}{N_i}+\sum_{k\sim i}\frac{a_i a_k}{N_i^2}-\sum_{k\sim i,  k\sim j}a_k \frac{a_i}{N_k^2}, 
\end{equation}
(where the last term is due to the triangles in the graph) when $j\sim i$, that is, $i$ and $j$ are nearest neighbours. When $i=j$, one gets
\begin{equation}\label{Jacobian2}
\frac{\partial f_i(\vect{a})}{\partial a_i}
=\sum_{k\sim i}\frac{a_k}{N_k}-a_i\sum_{k\sim i}\frac{a_k}{N_k^2}
-\frac{\sum_{k\sim i}a_k}{N_i}.
\end{equation}
The remaining non-vanishing partial derivatives correspond to
nodes $j$ located at distance 2 of $i$ in the graph, that is, to
nodes $j$ such that $j\sim k$ for some $k\sim i$, $j\ne i$ but $i\not \sim j$. Then
\begin{equation}\label{Jacobian3}
\frac{\partial f_i(\vect{a})}{\partial a_j}
=-\sum_{j\sim k,\ k\sim i}\frac{a_i a_k}{N_k^2}.
\end{equation}
When $N_i = N$, $\forall i$, these expressions simplify to 
$$\frac{\partial f_i(\vect{a})}{\partial a_j}
=\sum_{k\sim i}\frac{a_i a_k}{N_i^2}=\frac{N-\kappa}{N^2}a_i-\frac{a_i}{N^2}\sum_{k\sim i, k\sim j}a_k. $$
If $j\sim i$,
$$\frac{\partial f_i(\vect{a})}{\partial a_i}=-\frac{N-\kappa}{N^2}a_i,$$
and
$$\frac{\partial f_i(\vect{a})}{\partial a_j}
=-\sum_{k\sim i, k\sim j}\frac{a_i a_k}{N_k^2}=-\frac{a_i}{N^2}\sum_{k\sim i, k\sim j}a_k, $$
if $j\sim k$ for some $k\sim i$, $j\ne i$ but $i\not \sim j$.\\
Consider the sub-matrix $L$ given by $L=(\partial f_i(\vect{a})/\partial a_j)_{j\sim i}$.  Let $d(\vect{a})$ be the diagonal matrix of diagonal given by $\vect{a}$.
 The perturbation associated with the triangles contained in the graph is represented by the term $-\frac{a_i}{N^2}\sum_{k\sim i, k\sim j}a_k$ in $\frac{\partial f_i(\vect{a})}{\partial a_j}$ for $j\sim i$, and the related matrix is given by
\begin{eqnarray*}
\left(-\frac{a_i}{N^2}\sum_{k\sim i, k\sim j}a_k\right)\gamma_{ij}
&=& \left(-\frac{a_i}{N^2}\sum_{k}\gamma_{ik}a_k\gamma_{kj}\right)\gamma_{ij}\\
&=&\left( -\frac{1}{N^2}\left(d(a)\Gamma d(a)\Gamma-{\rm diag}( d(a)\Gamma d(a)\Gamma)\right)_{ij}\right)\gamma_{ij}\\
&=&\left( -\frac{1}{N^2}(d(a)\Gamma d(a)\Gamma)_{ij}+\frac{N-\kappa}{N^2}d(a)_{ij}\right)\gamma_{ij}.\\
\end{eqnarray*}
The matrix $L$ is now given by
$$L=\frac{d(a)}{N^2}(N-\kappa)(\Gamma-id)- \frac{1}{N^2}(d(a)\Gamma d(a) \Gamma-(N-\kappa)d(a))\circ\Gamma ,$$
where $\circ$ represents the Hadamard product, i.e. the multiplication component by component.\\
Likewise, the perturbation of $L$ by $\left(\frac{\partial f_i(\vect{a})}{\partial a_j}\right)_{i\sim k, k\sim j, i\not \sim j, i\neq j}$ can be written as
\begin{eqnarray*}
\left(-\frac{a_i}{N^2}\sum_{\substack{k\sim i, k\sim j, \\i\not \sim j, i\neq j}} a_k \right)\gamma_{ij}&=& \left(-\frac{a_i}{N^2}\sum_{k}\gamma_{ik}a_k\gamma_{kj}\right)(1-\gamma_{ij}-id_{ij})\\
&=& \left( -\frac{1}{N^2}(d(a)\Gamma d(a)\Gamma)_{ij}+\frac{N-\kappa}{N^2}d(a)_{ij}\right)(1-\gamma_{ij}-id_{ij}).\\
\end{eqnarray*}
The related Jacobian is thus given by  $L+\left(\frac{\partial f_i(\vect{a})}{\partial a_j}\right)_{i\sim k, k\sim j, i\not \sim j, i\neq j}$,
that is
\begin{eqnarray*}
df(a)&=&\frac{d(a)}{N^2}(N-\kappa)(\Gamma-id)- \frac{1}{N^2}(d(a)\Gamma d(a) \Gamma-(N-\kappa)d(a)) \circ \Gamma\\
&  &- \frac{1}{N^2}(d(a)\Gamma d(a) \Gamma-(N-\kappa)d(a)) \circ(\bold{1}-\Gamma-id)\\
&=& \frac{d(a)}{N^2}(N-\kappa)(\Gamma-id)- \frac{1}{N^2}(d(a)\Gamma d(a) \Gamma-(N-\kappa)d(a)) \circ (\bold{1}-id)\\
&=& \frac{d(a)}{N^2}(N-\kappa)(\Gamma-id)- \frac{1}{N^2}(d(a)\Gamma d(a) \Gamma-(N-\kappa)d(a)),\\
\end{eqnarray*}
where $\bold{1}$ is the matrix composed only of ones. The last equality is a consequence of  the fact that the diagonal of $d(a)\Gamma d(a) \Gamma-(N-\kappa)d(a)$ vanishes. Hence,
$$df(a)=\frac{d(a)\Gamma}{N^2}((N-\kappa)id-d(a)\Gamma ) =\frac{d(a)\Gamma}{N^2}( c\ id-d(a)\Gamma ), $$
proving the result.

 \subsection{Proof of Proposition \ref{StabilityBoundary}}
Set $I=\{i \in \Lambda : a_i=0\}$,  and 
consider the sub-graphs $\gamma_p$ of $G$ induced by the nodes of $J=\Lambda\setminus I$, with 
$\gamma_p=(\Lambda_p,E_p)$,  $1 \leq p \leq P$. 
The related  critical points $\vect{a}$ are such that the restrictions $\vect{a}\vert_{\gamma_p}$
satisfy the linear systems
 $\Gamma_{\gamma_p}\vect{a}\vert_{\gamma_p}=c_{\gamma_p} \vect{1}\vert_{\gamma_p}$. Set $N_{\gamma_p}=c_{\gamma_p}+ \kappa$.

\noindent  \eqref{Jacobian1} - \eqref{Jacobian3} permit to compute the entries of the Jacobian matrix,
by first looking at the diagonal entries:
When $i\in \Lambda_p$, one has
$$
\frac{\partial f_i(\vect{a})}{\partial a_i}
= - a_i \frac{N_{\gamma_p}-\kappa}{N_{\gamma_p}^2},
$$
providing  the diagonal entry of the Jacobian of $f\vert_{\gamma_p}(\vect{a}\vert_{\gamma_p})$.
When $i\not\in \Lambda_p$, a similar computation yields
$$
\frac{\partial f_i(\vect{a})}{\partial a_i}
= \sum_{k\sim i} \frac{a_k}{N_k} - \frac{N_i-\kappa}{N_i}.
$$
We then compute the entries $(i,j)$ for $j \sim i$:
$$
\frac{\partial f_i(\vect{a})}{\partial a_j}
=a_i \frac{N_{\gamma_p} - \kappa}{N_{\gamma_p}^2}-\sum_{k\sim i,  k\sim j} a_k \frac{a_i}{N_k^2} = a_i \frac{N_{\gamma_p} - \kappa}{N_{\gamma_p}^2}-\sum_{k\sim i,  k\sim j, k \in \Lambda_p}  \frac{a_k a_i}{N_{\gamma_p}^2}  , 
$$
for $i,j \in \Lambda_p$ and $1\leq p \leq P$, which corresponds to the $(i,j)$ entry of the Jacobian of $f\vert_{\gamma_p}(\vect{a}\vert_{\gamma_p})$.
Likewise,
$$
\frac{\partial f_i(\vect{a})}{\partial a_j}
= \frac{a_i}{N_j}- \frac{a_i}{N_{\gamma_p}} + \sum_{k\sim i} \frac{a_i a_k}{N_{\gamma_p}^2}-\sum_{\substack{k\sim i,  k\sim j, \\ k \in \Lambda_p}} \frac{a_k a_i}{N_{\gamma_p}^2} = \frac{a_i}{N_j} - a_i \frac{\kappa}{N_{\gamma_p}^2} - \frac{a_i}{N_{\gamma_p}^2} \sum_{\substack{k\sim i,  k\sim j,\\ k \in \Lambda_p}} a_k , 
$$
when $i \in \Lambda_p$ for some $p$ and $j \not \in \Lambda_p$.
 Finally,
$$
\frac{\partial f_i(\vect{a})}{\partial a_j}=0, 
$$
when $i,j \not \in \cup_p\Lambda_p$, or equivalently when both $i$ and $j$ belongs to $I$.

We next consider  $(i,j)$ entries where $j$ is at a distance 2 of $i$ in the graph $G$, that is  when $j$ is such that $j\sim k$ for some $k\sim i$, $j\neq i$ and $j \not\sim i$. One obtains that
$$
\frac{\partial f_i(\vect{a})}{\partial a_j}=-a_i \sum_{j \sim k, k \sim i} \frac{a_k}{N_{\gamma_p}^2}, 
$$
when $i,j,k \in \Lambda_p$,
which is  the $(i,j)$ entry of the Jacobian of $f\vert_{\gamma_p}(\vect{a}\vert_{\gamma_p})$.

 Likewise,
$$
\frac{\partial f_i(\vect{a})}{\partial a_j}=-a_i \sum_{j \sim k, k \sim i} \frac{a_k}{N_{\gamma_p}^2} , 
$$
when $i,k \in \Lambda_p, j\not \in \Lambda_p$ ($\Rightarrow j \in I$).

 Next,
$$
\frac{\partial f_i(\vect{a})}{\partial a_j}=0.
$$
when $i \, \mathrm{or} \, k \not \in \Lambda_p, \forall j \in \Lambda$.

Permuting conveniently the indices, the Jacobian ${\rm d}f(\vect{a})$ can be written as
\begin{equation}
{\rm d}f(\vect{a}) = \begin{pmatrix} d_n & \vect{0} \\ \ast & {\rm d}f^{\gamma} \end{pmatrix}
\end{equation}
where $d_n$ is a diagonal matrix $n \times n$ with entries given  by $\lambda_i:= \sum_{k\sim i} \frac{a_k}{N_k} - \frac{N_i-\kappa}{N_i}$, for $i \in I$,  and hence  ${\rm d}f^{\gamma}$ is a block diagonal matrix, each block being equal to the Jacobian of $f$ restricted on each sub-graph $\gamma_p$. The permutation allows us to group all indices $i \in I$ in the same block, and  all indices related to the sub-graphs $\gamma_p$  are also arranged together.
It follows  that the eigenvalues of ${\rm d}f(\vect{a})$ are given by the diagonal entries $(\lambda_i)_{i\in I}$, and by the eigenvalues of all Jacobian matrices.

{\bf Acknowledgements} This work was supported by the University of Fribourg, and by the
SystemsX  "Plant growth in changing environments" project funding. Many thanks to D.
Kierzkowski and C. Kuhlemeier for providing us the picture given in Figure
\ref{phyllo} and to Ale\v s Janka for its help in Matlab programming.
 We are very grateful to
Patrick Favre and Didier Reinhardt for giving us the opportunity to learn parts of the actual knowledge 
on the role of the auxin flux in plant patterning.

\end{document}